\documentclass{llncs}
\usepackage{graphicx}
\usepackage[labelformat=simple]{subcaption}
\usepackage{url}
\usepackage{amsmath}
\usepackage{amssymb}
\usepackage{amsfonts}
\usepackage{enumerate}
\usepackage[textsize=footnotesize]{todonotes}
\usepackage{todonotes}
\usepackage{xspace}
\graphicspath{{figs/}}
\newcommand{\bi}{\begin{itemize}}
\newcommand{\ei}{\end{itemize}}

\newenvironment{sketch}
{{\em Proof sketch:}}{\par\vspace{2mm}}

\newcommand\snip[1]{}

\renewcommand{\leq}{\leqslant} %needs \usepackage{amssymb}
\renewcommand{\geq}{\geqslant}

\newcommand{\R}{\mathcal{R}}
\newcommand{\br}{2.5D-BR}
\newcommand{\gbr}{2.5D-GBR}
\newcommand{\nc}{NC}

\makeatletter
\def\cramped                           % "Cramped" list style.
 {
\parskip0pt\@topsep0pt       % put after \begin{itemize}
  \itemsep0pt\parsep0pt
  \settowidth{\labelwidth}{\@itemlabel}
  \advance\leftmargin-\labelsep
  \advance\leftmargin-\labelwidth
  \parshape1\@totalleftmargin\linewidth
}
\makeatother

\title{Visibility Representations of Boxes in 2.5 Dimensions\thanks{%
Research  started at the 2016 Bertinoro
workshop on Graph Drawing. Research  supported in part 
by NSERC, and by MIUR project AMANDA prot. 2012C4E3KT\_001.}}
\author{
Alessio Arleo\inst{1} \and
Carla Binucci\inst{1} \and
Emilio Di Giacomo\inst{1} \and
William S. Evans\inst{2} \and
Luca Grilli\inst{1} \and
Giuseppe Liotta\inst{1} \and
Henk Meijer\inst{3}\and
Fabrizio Montecchiani\inst{1} \and \\
Sue Whitesides\inst{4} \and
Stephen Wismath\inst{5}
}
\institute{
Universit\`a degli Studi di Perugia, Italy, \email{alessio.arleo@studenti.unipg.it} \email{\{carla.binucci, emilio.digiacomo, luca.grilli, giuseppe.liotta, fabrizio.montecchiani\}@unipg.it} \and
University of British Columbia, Canada, \email{will@cs.ubc.ca}\and
University College Roosevelt, the Netherlands, \email{h.meijer@ucr.nl}\and 
University of Victoria, Canada, \email{sue@uvic.ca}\and
University of Lethbridge, Canada, \email{wismath@uleth.ca}   
}

\begin{document}

\pagenumbering{arabic}
\pagestyle{plain}

\maketitle

\begin{abstract}
We initiate the study of 2.5D box visibility representations (\br{}) where vertices are mapped to 3D boxes having the bottom face in the plane $z=0$ and edges are unobstructed lines of sight parallel to the $x$- or $y$-axis. We prove that: $(i)$ Every complete bipartite graph admits a \br{}; $(ii)$ The complete graph $K_n$ admits a \br{} if and only if $n \leq 19$; $(iii)$ Every graph with pathwidth at most $7$ admits a \br{}, which can be computed in linear time. We then turn our attention to 2.5D grid box representations (\gbr{}) which are \br{}s such that the bottom face of every box is a unit square at integer coordinates. We show that an $n$-vertex graph that admits a \gbr{} has at most $4n - 6 \sqrt{n}$ edges and this bound is tight. Finally, we prove that deciding whether a given graph $G$ admits a \gbr{} with a given footprint is NP-complete. The footprint of a \br{} $\Gamma$ is the set of bottom faces of the boxes in $\Gamma$. 
\end{abstract}

\section{Introduction}

A \emph{visibility representation (VR)} of a graph $G$ maps the vertices of $G$ to non-overlapping geometric objects  and the edges of $G$ to \emph{visibilities}, i.e., segments that do not intersect any geometric object other than at their end-points. Depending on the type of geometric objects representing the vertices and on the rules used for the visibilities, different types of representations have been studied in computational geometry and graph drawing.

%In a \emph{bar visibility representation (BVR)} the vertices are mapped to horizontal segments in the plane, called \emph{bars}, while visibilities are vertical segments. BVRs were introduced in the 80s as a modeling tool for VLSI problems~\cite{Duchet1983319,ov-grild-78,DBLP:journals/dcg/RosenstiehlT86,tt-rgc-91,t-prg-84,DBLP:conf/compgeom/Wismath85} and were investigated under various models. It is clear that in any such model the resulting visibility graph must be planar and indeed all planar graphs admits a weak BVR~\cite{}. Graphs that admit a strong BVR under the so called $\epsilon$-visibility model, have been characterized independently by Wismath~\cite{DBLP:conf/compgeom/Wismath85} and by Tamassia and Tollis~\cite{tt-rgc-91}. In the $\epsilon$-visibility model the segments representing the edges can be replaced by strips of non-zero width, i.e., the intersection point of a visibility and a bar is never an endpoint of the bar. Finally, a characterization of the graphs that admits a strong BVR (not assuming the $\epsilon$-visibility model) is given by Tamassia and Tollis~\cite{tt-rgc-91}.
    
A \emph{bar visibility representation (BVR)} maps the vertices to horizontal segments, called \emph{bars}, while visibilities are vertical segments. BVRs were introduced in the 80s as a modeling tool for VLSI problems~\cite{Duchet1983319,ov-grild-78,DBLP:journals/dcg/RosenstiehlT86,tt-rgc-91,t-prg-84,DBLP:conf/compgeom/Wismath85}. The graphs that admit a BVR are planar and they have been characterized under various models~\cite{Duchet1983319,DBLP:journals/dcg/RosenstiehlT86,tt-rgc-91,DBLP:conf/compgeom/Wismath85}. 

Extensions and generalizations of BVRs have been proposed in order to enlarge the family of representable graphs. In a \emph{rectangle visibility representation (RVR)} the vertices are axis-aligned rectangles, while visibilities are both horizontal or vertical segments~\cite{SoCG,DBLP:conf/gd/BoseDHS96,DBLP:journals/combinatorics/DeanEHP08,DBLP:journals/dam/DeanH97,JGAA-11,Hutchinson1999161,Shermer96,DBLP:conf/stacs/StreinuW03}. RVRs can exist only for graphs with thickness at most two and with at most $6n-20$ edges~\cite{Hutchinson1999161}. Recognizing these graphs is NP-hard in general~\cite{Shermer96} and can be done in polynomial time in some restricted cases~\cite{SoCG,DBLP:conf/stacs/StreinuW03}. Generalizations of RVRs where orthogonal shapes other than rectangles are used to represent the vertices have been recently proposed~\cite{ddelmmw-opvreg-16,DBLP:journals/ipl/LiottaM16}. 
Another generalization of BVRs are \emph{bar $k$-visibility representations ($k$-BVRs)}, where  each visibility segment can ``see'' through at most $k$ bars. Dean et al.~\cite{JGAA-136} proved that the graphs admitting a $1$-BVR have at most $6n-20$ edges. Felsner and Massow~\cite{JGAA-157} showed that there exist graphs with a $1$-BVR whose thickness is three. The relationship between $1$-BVRs  and  $1$-planar graphs has also been investigated~\cite{6777827,JGAA-330,JGAA-343,Sultana2014}.

RVRs are extended to 3D space by  \emph{Z-parallel Visibility Representations (ZPR)}, where vertices are axis-aligned rectangles belonging to planes parallel to the $xy$-plane, while visibilities are parallel to the $z$-axis. Bose et al.~\cite{JGAA-6} proved that $K_{22}$ admits a ZPR, while $K_{56}$ does not. {\v{S}}tola \cite{Stola2009} subsequently reduced the upper bound on the size of the largest representable complete graph by showing that $K_{51}$ does not admits a ZPR. Fekete et al.~\cite{Fekete1995} showed that $K_7$ is the largest complete graph that admits a ZPR if unit squares are used to represent the vertices.    
A different extension of RVRs to 3D space are the \emph{box visibility representations (BR)} where vertices are 3D boxes, while visibilities are parallel to the $x$-, $y$- and $z$- axis. This model was studied by Fekete and Meijer~\cite{DBLP:journals/ijcga/FeketeM99} who proved that $K_{56}$ admits a BR, while  $K_{184}$ does not.

In this paper we introduce \emph{2.5D box visibility representations (\br{})} where vertices are 3D boxes whose bottom faces lie in the plane $z=0$ and visibilities are parallel to the $x$- and $y$-axis. Like the other 3D models that use the third dimension, \br{}s overcome some limitations of the 2D models. For example, graphs with arbitrary thickness can be realized.
%On the other hand, because of the requirement that the bottom faces of all boxes are in a common plane, this model is less powerful than other three-dimensional models. For example, the largest complete graph that can be realized is smaller that those that admit a ZPR or a BR. 
In addition \br{}s seem to be simpler than other 3D models from a visual complexity point of view and have the advantage that they can be physically realized, for example by 3D printers or by using physical boxes. Furthermore, this type of representation can be used to model visibility between buildings of a urban area~\cite{Carmi2015251}.
The main results of this paper are as follows.  
  
\begin{itemize}
\item We show that every complete bipartite graph admits a \br{} (Section~\ref{se:complete}). This implies that there exist graphs that admit a \br{} and have arbitrary thickness.

\item We prove that the complete graph $K_n$ admits a \br{} if and only if $n \leq 19$ (Section~\ref{se:complete}). Thus, every graph with $n \leq 19$ vertices admits a \br{}. 

\item We describe a technique to construct a \br{} of every graph with pathwidth at most $7$, which can be computed in linear time (Section~\ref{se:pathwidth}). 

\item We then study \emph{2.5D grid box representations} (\emph{\gbr{}}) which are \br{}s such that the bottom face of every box is a unit square with corners at integer coordinates (Section~\ref{se:spiderman}). We show that an $n$-vertex graph that admits a \gbr{} has at most $4n - 6 \sqrt{n}$ edges and that this bound is tight. It is worth remarking that VRs where vertices are represented with a limited number of shapes have been investigated in the various models of visibility representations. Examples of these shape-restricted VRs are unit bar VRs~\cite{dv-ubvg-03}, unit square VRs~\cite{DBLP:journals/combinatorics/DeanEHP08}, and unit box VRs~\cite{DBLP:journals/ijcga/FeketeM99}. 

\item Finally, we prove that deciding whether a given graph $G$ admits a \gbr{} with a given footprint is NP-complete (Section~\ref{se:spiderman}). The \emph{footprint} of a \br{} $\Gamma$ is the set of bottom faces of the boxes in $\Gamma$. 
\end{itemize} 

\noindent For reasons of space, some proofs and details are omitted and can be found in the appendix.

\section{Preliminaries}\label{se:preliminaries}

A \emph{box} is a six-sided polyhedron of non-zero volume with axis-aligned sides in a 3D Cartesian coordinate system. In a \emph{2.5D box representation} (\emph{\br{}}) the vertices are mapped to boxes that lie in the non-negative half space $z \geq 0$ and include one face in the plane $z=0$, while each edge is mapped to a \emph{visibility} (i.e. a segment whose endpoints lie in faces of distinct boxes and whose interior does not intersect any box) parallel to the $x$- or to the $y$-axis. We remark that visibilities between non-adjacent objects may exist, i.e., we adopt the so called \emph{weak visibility model} (in the \emph{strong visibility model} each visibility between two geometric objects corresponds to an edge of the graph). The weak model seems to be the most effective when representing non-planar graphs and it has been adopted in several works (see e.g.~\cite{SoCG,JGAA-330,JGAA-343}). 
As in many papers on visibility representations~\cite{DBLP:journals/ijcga/FeketeM99,DBLP:journals/ipl/KantLTT97,DBLP:conf/stacs/StreinuW03,TamassiaTollis86,DBLP:conf/compgeom/Wismath85}, we assume the \emph{$\epsilon$-visibility model}, where each segment representing an edge is the axis of a positive-volume cylinder that intersects no box except at its ends; this implies that an intersection point between a visibility and a box belongs to the interior of a box face. 
In what follows, when this leads to no confusion, we shall use the term \emph{edge} to indicate both an edge and the corresponding visibility, and the term \emph{vertex} for both a vertex and the corresponding geometric object.

Given a box $b$ of a \br{}, the face that lies in the plane $z=0$ is called the \emph{footprint} of $b$. The intersection of the plane $z=0$ with a \br{} $\Gamma$ is called the \emph{footprint} of $\Gamma$ and is denoted by $\Gamma_0$.
In other words, the footprint of a \br{} $\Gamma$ consists of the footprint of all the boxes in $\Gamma$. 
If $\Gamma$ is a \br{} of a complete graph then its footprint $\Gamma_0$ satisfies a trivial necessary condition (throughout the paper we will refer to this condition as \emph{\nc{}}): for every pair of boxes $b_1$ and $b_2$ of $\Gamma$, there must exist a line $\ell$ (in the plane $z = 0$) such that (\emph{i}) $\ell$ passes through the footprints of $b_1$ and $b_2$, and ($ii$) $\ell$ is either parallel to the $x$-axis or to the $y$-axis.
A \emph{2.5D grid box representation} (\emph{\gbr{}}) is a \br{} such that every box has a footprint that is a unit square with corners at integer coordinates.

Two boxes \emph{see each other} if there exists a visibility between them; we say that they \emph{see each other above another box} $b$, if there exists a visibility between them and the projection of this visibility on the plane $z=0$ intersects the interior of the footprint of $b$. Notice that this implies that the two boxes are both taller than $b$. We say that two boxes have a \emph{ground visibility} or are \emph{ground visible} if there exists a visibility between their footprints, i.e. if there exists an unobstructed axis-aligned line segment connecting their footprints. If two boxes are ground visible then they see each other regardless of their heights and the heights of the other boxes. 
Let $G$ be a graph, let $\Lambda$ be a collection of boxes each lying in the non-negative half space $z \geq 0$ with one face in the plane $z = 0$, such that the boxes of $\Lambda$ are in bijection with the vertices of $G$. Note that $\Lambda$ may not be a \br{} of $G$. For a vertex $v$ of $G$, $\Lambda(v)$ denotes the corresponding box in $\Lambda$, while $h(\Lambda(v))$, or simply $h(v)$, indicates the height of this box. For a subset $S \subset V(G)$, $\Lambda(S)$  denotes the subset of boxes associated with $S$, while $\Lambda_0(S)$ is the footprint of $\Lambda(S)$. Let $G[S]$ be the subgraph of $G$ induced by $S$. We say that $\Lambda(S)$ is a \br{} of $G[S]$ in $\Lambda$, if for any edge $(u,v)$ of $G[S]$ there exists a visibility in $\Lambda$ between $\Lambda(u)$ and $\Lambda(v)$; that is, the visibility is not destroyed by the presence of the other boxes in $\Lambda$.

%Finally, we say that two boxes of $\Lambda(S)$ are ground visible in $\Lambda$, if their footprints can see each other in $\Lambda_0$.

% [Start Old-Emilio-version], commented by Luca
%
%In what follows we will say that two boxes \emph{see each other} if there exists a visibility between them; we will say that they see each other above another box $B$, if there exists a visibility between them only if they both have height larger than the height of $B$. We will say that two boxes have a \emph{ground visibility} if there exists a visibility between their bottom faces. It is immediate to see that if two boxes have a ground visibility then they see each other whatever are their heights.
%
% [End Old-Emilio-version], commented by Luca

%--------------------------------------------------------------
\section{2.5D Box Representations of Complete Graphs}\label{se:complete}

In this section we consider \br{}s of complete graphs and complete bipartite graphs. 

%\begin{figure}[t]
%\centering
%\includegraphics[scale=0.6]{bipartite}
%\caption{Construction of a \br{} of a complete bipartite graph $K_{m,n}$. The numbers in the rectangles indicate the height of each box.}
%\label{fig:bipartite}
%\end{figure}

\begin{theorem}\label{th:bipartite}
Every complete bipartite graph admits a \br{}.
\end{theorem}
\begin{proof}
Let $K_{m,n}$ be a complete bipartite graph. We represent the $m$ vertices in the first partite set with $m$ boxes $a_0, a_1,$ $\ldots, a_{m-1}$ such that box $a_i$ has a footprint with corners at ($2i,0,0$),  ($2i+1,0,0$),  ($2i,2n-1,0$) and  ($2i+1,2n-1,0$) and height $m-i$.
Then we represent the $n$ vertices in the second partite set with $n$ boxes $b_0,b_1,\ldots,b_{n-1}$ such that box $b_j$ has a footprint with corners at ($2m,2j,0$),  ($2m+1,2j,0$),  ($2m,2j+1,0$) and ($2m+1,2j+1,0$) and height $m$. 
%See Figure~\ref{fig:bipartite} for an illustration. 
Consider now a box $a_i$ and a box $b_j$. By construction $a_i$ and $b_j$ see each other above all boxes $a_l$ with $l>i$.\qed
\end{proof}

A consequence of Theorem~\ref{th:bipartite} is that there exist graphs with unbounded thickness that admit a \br{}. This contrasts with other models of visibility representations (e.g., $k$-BVRs, and RVRs), which can only represent graphs with bounded thickness. 

\medskip

We now prove that the largest complete graph that admits a \br{} is $K_{19}$. We first show that given a \br{} of a complete graph there is one line parallel to the $x$-axis and one line parallel to the $y$-axis whose union intersect all boxes and such that each of them intersects at most $10$ boxes. This implies that there can be at most $20$ boxes in a \br{} of a complete graph. We then show that there must be a box that is intersected by both lines, thus lowering this bound to $19$. We finally exhibit a \br{} of $K_{19}$. We start with some technical lemmas. The proof of the next one can be found in Appendix~\ref{ap:complete}. 

\begin{lemma}\label{le:heights}
Let $G$ be an $n$-vertex graph that admits a \br{} $\Gamma'$. Then there exists a \br{} $\Gamma$ of $G$ such that every box of $\Gamma$ has a distinct integer height in the range $[1,n]$ and the footprint of $\Gamma$ is the same as that of $\Gamma'$.  
\end{lemma}

The following lemma is proved in~\cite[Obervation 1]{Kleitman2001}; we give an alternative proof in Appendix~\ref{ap:complete}. Given an axis-aligned rectangle $r$ in the plane $z=0$, we denote by $x(r)$ the $x$-extent of $r$ and by $y(r)$ the $y$-extent of $r$, so $r = x(r) \times y(r)$.  

\begin{lemma}\emph{\cite{Kleitman2001}}\label{le:cross}
For every arrangement $\R$ of $n$ axis-aligned rectangles in the plane
such that for all $a,b \in \R$, either $x(a) \cap x(b) \neq \emptyset$
or $y(a) \cap y(b) \neq \emptyset$, there exists a vertical and a
horizontal line whose union intersects all rectangles in~$\R$.
\end{lemma}

The following lemma is similar to the Erd\H{o}s--Szekeres lemma and can be proved in a similar manner~\cite{Fekete1995}. A sequence of distinct integers is \emph{unimaximal} if no element of the sequence is smaller than both its predecessor and successor.

%The following lemma is a consequence of the Erd\H{o}s--Szekeres theorem (see, e.g.,~\cite{Fekete1995}). A sequence of distinct integers is \emph{unimaximal} if it has exactly one local maximum.

\begin{lemma}\emph{\cite{Fekete1995}}\label{le:unimaximal}
For all $m > 1$, in every sequence of ${m \choose 2}+1$ distinct integers, there exists at least one unimaximal sequence of length $m$.
\end{lemma}

Given a \br{} $\Gamma$ and a line $\ell$ parallel to the $x$-axis or to the $y$-axis, we say that $\ell$ \emph{stabs} a set of boxes $B$ of $\Gamma$ if it intersects the interior of the footprints of each box in $B$. Let $b_1, b_2, \dots,b_h$ be the boxes of $B$ in the order they are stabbed by $\ell$. We say that $B$ has a \emph{staircase layout}, if $h(b_i) > h(b_{i-1})$ for $i=2,3,\dots,h$.

\begin{lemma}\label{le:unimaximal-5}
In a \br{} of a complete graph no line parallel to the $x$-axis or to the $y$-axis can stab five boxes whose heights, in the order in which the boxes  are stabbed, form a unimaximal sequence.
\end{lemma}
\begin{proof}
Assume, as a contradiction, that there exists a line $\ell$ parallel to the $x$-axis or to the $y$-axis that stabs $5$ boxes $b_1$, $\dots$, $b_5$ whose heights form a unimaximal sequence in the order in which the boxes are stabbed by $\ell$.  Let $r_{i}$ be the footprint of box $b_{i}$ (with $1 \leq i \leq 5$). We claim that there exists a ground visibility between every pair of boxes $b_i$ and $b_{j}$ (with $1 \leq i < j \leq 5$). If $j=i+1$ this is clearly true. Suppose then that $j \neq i+1$. If $b_i$ and $b_j$ do not have a ground visibility, then they must see each other above $b_l$ with $i < l < j$, i.e., the height of $b_i$ and of $b_j$ must be larger than the height of $b_l$, which is impossible because the sequence of heights is unimaximal. Thus, for every pair of boxes $b_i$ and $b_j$ there must be a ground visibility. Since $b_i$ and $b_j$ are both stabbed by $\ell$, this visibility must be parallel to $\ell$. This implies that the left sides (if $\ell$ is parallel to the $x$-axis) or the bottom sides (if $\ell$ is parallel to the $y$-axis) of rectangles $r_{1}, r_{2}, r_{3}, r_{4}, r_{5}$ form a bar visibility representation of $K_5$, which is impossible because bar visibility representations exist only for planar graphs~\cite{GareyJS76}.
\qed \end{proof}

\begin{lemma}\label{le:stabber10}
In a \br{} of a complete graph no line parallel to the $x$-axis or to the $y$-axis can stab more than $10$ boxes.
\end{lemma}
\begin{proof}
Let $\Gamma$ be a \br{} of a complete graph $K_n$. By Lemma~\ref{le:heights} we can assume that all boxes have distinct integer heights. Suppose, as a contradiction, that there exists a line $\ell$ parallel to the $x$-axis or to the $y$-axis that stabs $k>10$ boxes. Let $h_1, h_2, \dots, h_k$ be the heights of the stabbed boxes in the order in which the boxes are stabbed by $\ell$.  By Lemma~\ref{le:unimaximal} this sequence of heights contains a unimaximal sequence of length $5$, but this is impossible by Lemma~\ref{le:unimaximal-5}.
\qed \end{proof}

\begin{lemma}\label{le:k20}
A complete graph admits a \br{} only if it has at most $19$ vertices. 
\end{lemma}
\begin{proof}
Let $\Gamma$ be a \br{} of a complete graph $K_n$ (for some $n>0$). By Lemma~\ref{le:heights} we can assume that all boxes of $\Gamma$ have distinct heights.  The footprint $\Gamma_0$ of $\Gamma$ is an arrangement of rectangles that satisfies Lemma~\ref{le:cross}. Thus there exist a line $\ell_h$ parallel to the $x$-axis and a line $\ell_v$ parallel to the $y$-axis that together stab all boxes of $\Gamma$. By Lemma~\ref{le:stabber10}, both $\ell_h$ and $\ell_v$ can stab at most $10$ boxes each.  This means that the number of boxes (and therefore the number of vertices of $K_n$) is at most $20$. We now prove that if $\ell_h$ and $\ell_v$ both stab ten boxes, there must be one box that is stabbed by both $\ell_h$ and $\ell_v$, which implies that the number of boxes in $\Gamma$ is at most $19$.

Suppose, for a contradiction, that $p = \ell_h \cap \ell_v$ does not lie in a box. Refer to Figure~\ref{fi:k20-1} for an illustration. Denote by $T$ the set of boxes stabbed by $\ell_v$ that are above $p$ and by $B$ be the set of boxes stabbed by $\ell_v$ that are below $p$. Analogously, denote by $L$ the set of boxes stabbed by $\ell_h$ that are to the left of $p$ and by $R$ the set of boxes stabbed by $\ell_h$ that are to the right of $p$. Each of these sets can be empty but $|T|+|B|=10$ and $|L|+|R|=10$. Denote by $l_1, l_2, \dots, l_{|L|}$ the set of boxes in $L$ from right to left, i.e., $l_1$ is the box closest to $p$. Analogously, denote by $r_1, r_2, \dots, r_{|R|}$ the boxes of $R$ from left to right ($r_1$ is the closest to $p$), by $t_1, t_2, \dots, t_{|T|}$ the boxes of $T$ from bottom to top ($t_1$ is the closest to $p$) and by $b_1, b_2, \dots, b_{|B|}$ the boxes of $B$ from top to bottom ($b_1$ is the closest to $p$). Let $f_T$, $f_B$, $f_L$, and $f_R$ be the footprints of $t_1$, $b_1$, $l_1$, and $r_1$, respectively. Let $\ell_X$ be the line containing the side of $f_X$ that is closest to $p$ and let $\ell'_X$ be the line containing the opposite side of $f_X$ (for every $X \in \{T,B,L,R\}$). 

We first claim that for each $f_X$ there exists a line $\ell_Y$ (with $X,Y \in \{T,B,L,R\}$ and $Y \neq X$) that intersects the interior of $f_X$. Suppose, for a contradiction, that this is not true for at least one $f_X$, say $f_L$; that is, the interior of $f_L$ is not intersected by $\ell_T$ and $\ell_B$. If so, there must be a line $\ell$ parallel to the $y$-axis that intersects all the rectangles in $T \cup B$ and $f_L$; otherwise the necessary condition \nc{} does not hold for $T \cup B \cup \{l_1\}$. But then $\ell$ would stab eleven boxes, which is impossible by Lemma~\ref{le:stabber10}. Thus, our claim holds and the four rectangles $f_X$ are placed so that $\ell_T$, $\ell_R$,
$\ell_B$, and $\ell_L$ stab $f_R$, $f_B$, $f_L$, and $f_T$ (or, symmetrically, $f_L$, $f_T$, $f_R$, and $f_B$, which follows a symmetric argument), respectively, as in Figure~\ref{fi:k20-1}.

\begin{figure}[t]
\centering
\begin{subfigure}[b]{.3\linewidth}
\centering
\includegraphics[width=\columnwidth,page=1]{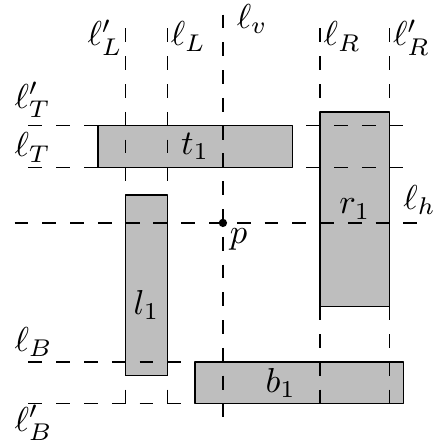}
\caption{}\label{fi:k20-1}
\end{subfigure}%
\hfil
\begin{subfigure}[b]{.3\linewidth}
\centering
\includegraphics[width=\columnwidth,page=2]{k20}
\caption{}\label{fi:k20-2}
\end{subfigure}
\hfil
\begin{subfigure}[b]{.3\linewidth}
\centering
\includegraphics[width=\columnwidth,page=3]{k20}
\caption{}\label{fi:k20-3}
\end{subfigure}
\caption{(a) Placement of the four rectangles $f_T$, $f_R$, $f_B$, and $f_L$. (b) Configuration A for the boxes of set $B$. (c) Configuration B for the boxes of set $B$. The arrow intersects the boxes that must have a staircase layout.}\label{fi:k20}
\end{figure}

We consider now the sets $T$, $B$, $L$, and $R$. For each set there are two possible configurations. Consider the set $B$ and the line $\ell'_L$. If the set $B' = B \setminus \{b_1\}$ contains a box $b_j$ whose footprint is completely to the right of $\ell'_L$, we say that $B$ has \emph{configuration A} (see Figure~\ref{fi:k20-2}). In the case of configuration A,  the footprint of all boxes in $L'=L \setminus \{l_1\}$ must extend below the line $\ell'_B$ (otherwise the necessary condition \nc{} does not hold for $L' \cup \{b_j\}$). This implies that $y(f_B)$ is contained in $y(l_i)$ for all $i \geq 2$. The only possibility for $b_1$ to see all these boxes is that $L'$ has a staircase layout (with $l_2$ being the shortest box) and $b_1$ is taller than the second tallest one. So, configuration A for the set $B$ implies that $L'$ has a staircase layout.  
If all boxes of $B'$ have a footprint that extends to the left of $\ell'_L$, we say that $B$ has \emph{configuration B} (see Figure~\ref{fi:k20-3}). In this case, $x(f_L)$ is contained in $x(b_i)$ for all $i \geq 2$. Again, the only possibility for $l_1$ to see all these boxes is that $B'$ has a staircase layout and that $l_1$ is taller than the second tallest one. So, configuration B for the set $B$ implies that $B'$ has a staircase layout. 
The definitions of configurations A and B for $T$, $L$, $R$ are similar to those for $B$ and arise by considering lines $\ell'_R$, $\ell'_T$, $\ell'_B$, respectively.

For any two sets $X$ and $Y$ that are consecutive in the cyclic order $T$, $R$, $B$, $L$, either $X'$ or $Y'$ has a staircase layout (depending on whether $X$ has configuration A or B). This implies that either $B'$ and $T'$ have both a staircase layout or $L'$ and $R'$ have both a staircase layout.   
%We now claim that either $B'$ and $T'$ have both a staircase layout or $L'$ and $R'$ have both a staircase layout.
%If $B$ and $T$ have configuration A, then $L'$ and $R'$ have a staircase layout. If they have configuration B, then $B'$ and $T'$ have a staircase layout. If $B$ has configuration A and $T$ has configuration B, then $L'$ and $T'$ have a staircase layout. Consider now $R$. If $R$ has configuration A, then $B'$ has a staircase layout, and therefore $B'$ and $T'$ have a staircase layout. If $R$ has configuration B, then $R'$ has a staircase layout, and therefore $L'$ and $R'$ have a staircase layout. A symmetric argument applies if $B$ has configuration B and $T$ has configuration A, and therefore, in all cases, either $B'$ and $T'$ have both a staircase layout or $L'$ and $R'$ have both a staircase layout, as claimed. 
Suppose that $B'$ and $T'$ have a staircase layout (the case when $L'$ and $R'$ have a staircase layout is analogous). If either $|B'| \geq 5$ or $|T'| \geq 5$, $\ell_v$ stabs at least five boxes whose heights form a unimaximal sequence, which is impossible by Lemma~\ref{le:unimaximal-5}. Thus $|B'|=4$ and $|T'|=4$ (recall that $|B'|+|T'|=8$).  Since all boxes of $\Gamma$ have distinct heights, either $h(b_2) < h(t_2)$ or $h(t_2) < h(b_2)$. In the first case $\ell_v$ stabs the five boxes $t_5, t_4, t_3, t_2, b_2$ whose heights form a unimaximal sequence, which is impossible by Lemma~\ref{le:unimaximal-5}. In the other case $\ell_v$ stabs the five boxes $b_5, b_4, b_3, b_2, t_2$ whose heights form a unimaximal sequence, which is impossible by Lemma~\ref{le:unimaximal-5}.\qed       
\end{proof}

We conclude this section by exhibiting a \br{} of $K_{19}$, illustrated in Figure~\ref{fig:K19}. To prove the correctness of the drawing the idea is to partition the vertex set of $K_{19}$ into five subsets (shown in Figure~\ref{fig:K19}) and prove that all boxes in a given set see all other boxes (details are in Appendix~\ref{ap:complete}). The following theorem holds.

\begin{figure}[t]
\centering
\includegraphics[width=0.8\textwidth]{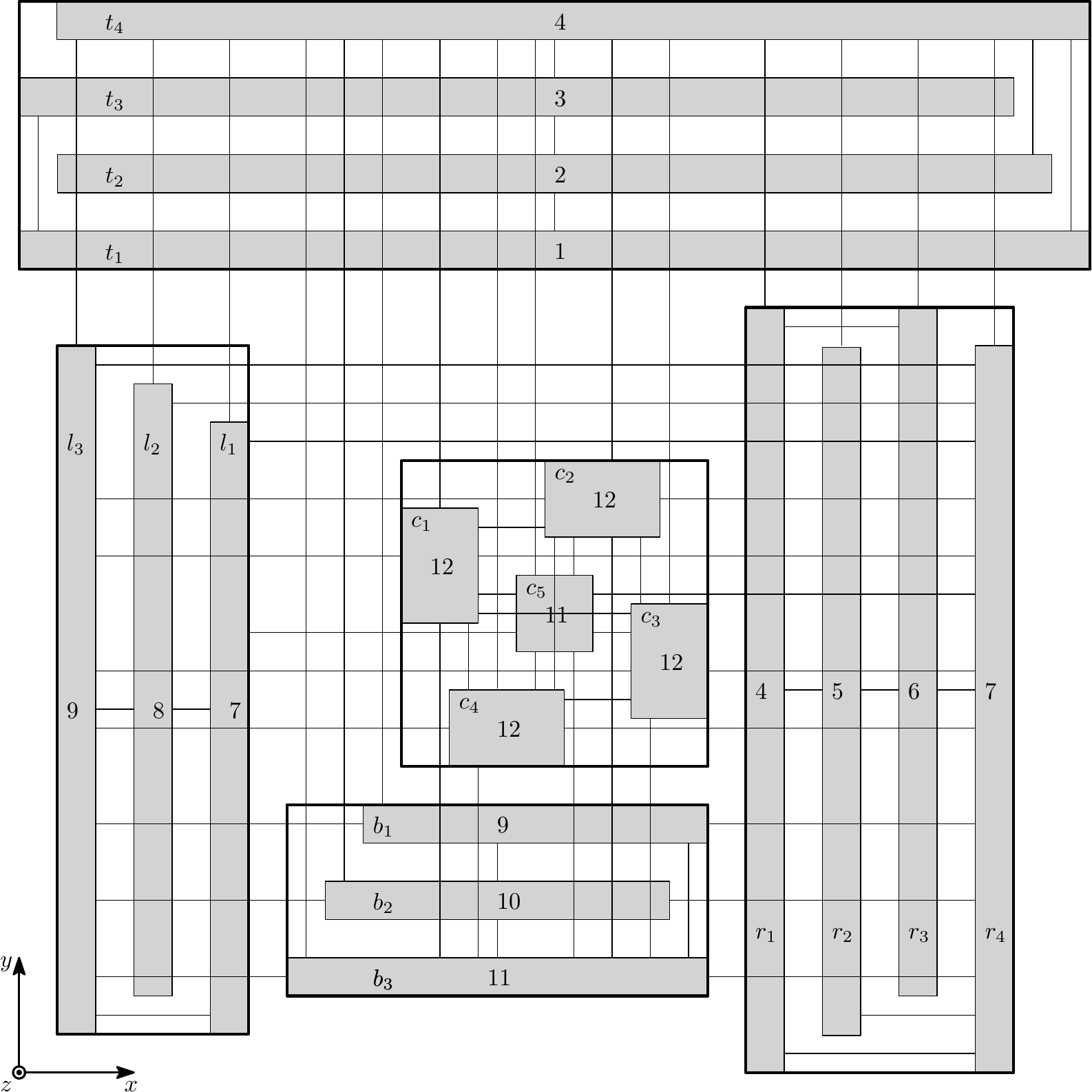}
\caption{\label{fig:K19}Illustration of a \br{} of $K_{19}$, the footprint is represented by a 2D drawing in the plane $z = 0$, while the heights of boxes are indicated by integer labels. The five rectangles with thick sides represent the partitioning of $V(K_{19})$ into five subsets.}
\end{figure}

\begin{theorem}
\label{th:complete}
A complete graph $K_n$ admits a \br{} if and only if $n \leq 19$. 
\end{theorem}

%\begin{figure}
%\label{fig:K19}
%\centering
%\includegraphics[width=10.0cm]{K19.pdf}
%\caption{$K_{19}$ represented as a set of boxes. $Z=0$ plane with heights indicated. Visibility projections color coded--  red: single visibilities; blue/purple: multiple horizontal visibilities right/left; green/black: multiple vertical visibilities down/up.}
%\end{figure}

\section{2.5D Box Representations of Graphs with Pathwidth at Most 7}\label{se:pathwidth}

A graph $G$ with pathwidth $p$ is a subgraph of a graph that can be constructed as follows. Start with the complete graph $K_{p+1}$ and classify  all its vertices as \emph{active}. At each step, a vertex is \emph{deactivated} and a new active vertex is introduced and joined to all the remaining active vertices. The order in which vertices are introduced is given by a \emph{normalized path decomposition}, which can be computed in linear time for a fixed $p$~\cite{Gupta2005242}. For a definition of normalized path decomposition see Appendix~\ref{ap:pathwidth}. 

\begin{figure}[t]
\centering
\includegraphics[width=0.9\textwidth, page=2]{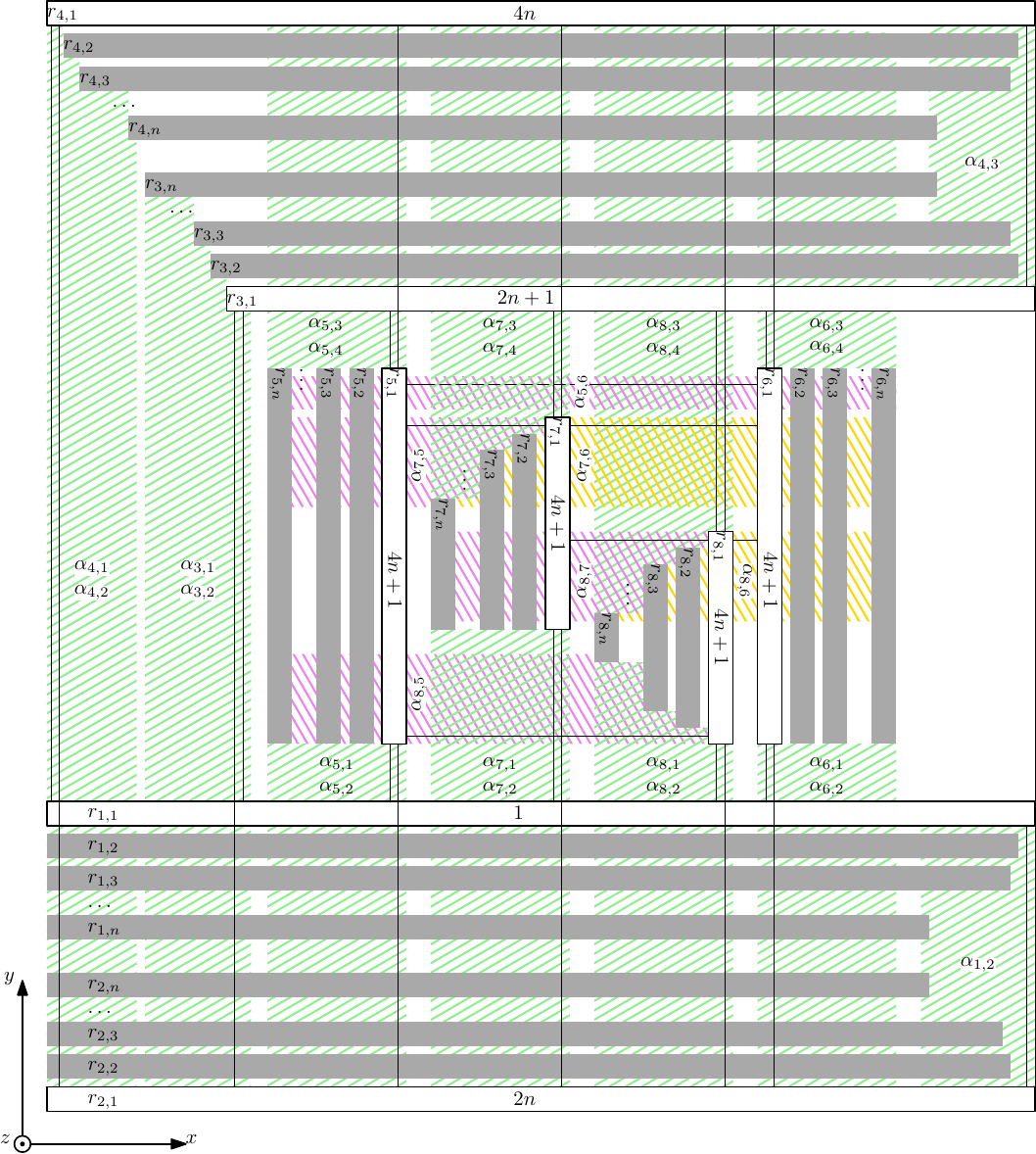}
\caption{Construction of a \br{} for a graph with pathwidth $7$. }
\label{fi:pathwdith}
\end{figure}

\begin{theorem}\label{th:pathwidth}
Every $n$-vertex graph with pathwidth at most $7$ admits a \br{}, which can be computed in $O(n)$ time.
\end{theorem}
\begin{proof}
We describe an algorithm to compute a \br{} of a graph $G$ with pathwidth $7$. The algorithm is based on the use of eight groups of rectangles, a subset of which will form the footprint of the \br{} of $G$. For graphs with pathwidth $p < 7$, the same algorithm can be applied by considering only $p+1$ groups, arbitrarily chosen.

The eight groups are defined in the plane $z=0$ and have $n$ rectangles each denoted as $r_{h,1}, r_{h,2}, \dots, r_{h,n}$ ($1 \leq h \leq 8$). The groups are placed as shown in Figure~\ref{fi:pathwdith}. The groups $h =5,6,7,8$ will be called \emph{central groups}. A vertex whose footprint is $r_{h,k}$ will be called a \emph{vertex of group $h$} ($1 \leq h \leq 8$). 

Let $v_1, v_2, \dots, v_n$ be the vertices of $G$ in the order given by a normalized path decomposition. We denote by $G_i$ the subgraph of $G$ induced by $\{v_1,v_2,\dots,v_i\}$. 
We create a collection of boxes by adding one box per step; at step $i$ we add a box to represent the next vertex $v_i$ to be activated. We denote the collection of the first $i$ boxes as $\Lambda_i$ and we prove that $\Lambda_i$ satisfies the following invariant (I1): $\Lambda_i$ is a \br{} of $G_i$ such that for any pair of boxes of group $j$ and $k$ ($1 \leq j,k \leq 8$) that represent vertices that are adjacent in $G_i$, there exists a visibility whose projection in the plane $z=0$ is inside the region $\alpha_{j,k}$. The regions 
$\alpha_{j,k}$ are highlighted in Figure~\ref{fi:pathwdith} as dashed regions.  

  The initial eight active vertices $v_1, v_2, \dots, v_8$ are represented by boxes whose footprints are $r_{1,1}, r_{2,1}, \dots, r_{8,1}$, respectively. The heights are set as follows: $h(v_h)=(h-1) \cdot n+1$, for $h=1,2,3,4$, and $h(v_h)=4n+1$ for $h=5,6,7,8$. The initial eight vertices are shown in Figure~\ref{fi:pathwdith} as white rectangles whose heights are shown inside them. $\Lambda_8$ satisfies invariant I1 thanks to the visibilities shown in Figure~\ref{fi:pathwdith}. 

Assume now that $\Lambda_{i-1}$ ($i>8$) satisfies invariant I1 and let $v_j$ be the vertex to be deactivated (for some $j < i$). Assume that $v_j$ belongs to group $h$ ($1 \leq h \leq 8$). Vertex $v_i$ is represented as a box with footprint $r_{h,i}$ and height $h(v_i)=h(v_j)+1$, if $h \in \{1,3,5,6,7,8\}$, or $h(v_i)=h(v_j)-1$, if $h \in \{2,4\}$. If the group of $v_i$ is a central group, we increase by one unit the height of all the active vertices of the other central groups. Notice that the heights of the vertices of group $h$, for $h \leq 4$, are in the range $[(h-1) \cdot n+1,h \cdot n]$, while the heights of the remaining vertices are greater than $4n$.   

We now prove that $\Lambda_i$ satisfies invariant I1 by showing that the addition of $v_i$ does not destroy any existing visibility and that $\Lambda_i(v_i)$ sees all the other active vertices inside the appropriate regions. We have different cases depending on the group $h$ of $v_i$. 

\noindent -- \textbf{$h=1$ or $h=2$.} The box $\Lambda_i(v_i)$ only intersects the regions $\alpha_{h',2}$, with $h' \neq 2$. Thus, the only visibilities that could be destroyed are those inside these regions. The visibilities in the regions $\alpha_{3,2}$, $\alpha_{4,2}$, $\alpha_{5,2}$,  $\alpha_{6,2}$, $\alpha_{7,2}$, and $\alpha_{8,2}$ are not destroyed by the addition of $v_i$ because the boxes representing the vertices of group $2$ are taller than the box representing $v_i$ and so are the boxes of any group $h'$ with $h'>2$. The existing visibilities in the region $\alpha_{1,2}$ are not destroyed because $r_{h,i}$ is short enough (in the $x$-direction) so that the existing boxes of groups $1$ and $2$ can still see each other in region $\alpha_{1,2}$. So, no visibility is destroyed for the vertices of group $2$. The box $\Lambda_i(v_i)$ sees the box of the active vertex of group $1$ or $2$ via a ground visibility in region $\alpha_{1,2}$ and it sees the boxes of all the other active vertices inside the region $\alpha_{h',1}$, with $h' > 2$, above the boxes of group $1$ (which are all shorter than it).

\noindent -- \textbf{$h=3$ or $h=4$.} The proof of this case can be found in Appendix~\ref{ap:pathwidth}. 

\noindent -- \textbf{$h=5$ or $h=6$.}  The box $\Lambda_i(v_i)$ only intersect the regions $\alpha_{h,h'}$, with $h' \in \{5,6,7,8\}$ and $h' \neq h$. However, it does  not intersect any existing visibility inside these regions and therefore the addition of $\Lambda_i(v_i)$ does not destroy any existing visibility. The box $\Lambda_i(v_i)$ sees the active vertices of groups $1$ and $2$ inside $\alpha_{h,k}$ (with $h=5$ or $6$, and $k = 1,2$) and above the boxes of group $1$. The active vertices of groups $3$ and $4$ are seen inside $\alpha_{h,k}$ (with $h=5$ or $6$, and $k = 3,4$) and above the boxes of group $3$. Finally, the active vertices of the central groups are seen inside $\alpha_{h,k}$ (with $h=5$ or $6$, and $k > 4$) and above the boxes of group $h$. Recall that the active vertices of the central groups have been raised to have the same height as $\Lambda_i(v_i)$ (which is larger than the height of any other box in the central groups). 

\noindent -- \textbf{$h=7$ or $h=8$.} The proof of this case can be found in Appendix~\ref{ap:pathwidth}.

The above construction can be done in $O(n)$ time. Since the normalized path decomposition can be computed in $O(n)$ time, the time complexity follows.\qed
\end{proof}

\section{2.5D Grid Box Representations}\label{se:spiderman}

Next we give a tight bound on the edge density of graphs admitting a \gbr{}. The proof, which appears in Appendix~\ref{ap:spiderman}, is based on the fact that a set of aligned (unit square) boxes induces an outerplanar graph. A square grid of boxes gives the bound.

\begin{theorem}\label{th:density}
Every $n$-vertex graph that admits a \gbr{} has at most $4n-6\sqrt{n}$ edges, and this bound is tight.
\end{theorem}

In the next theorem we prove that deciding whether a given graph admits a \gbr{} with a given footprint is NP-complete. We call this problem \gbr{}-WITH-GIVEN-FOOTPRINT (2.5GBR-WGF). The reduction is from HAMILTONIAN-PATH-FOR-CUBIC-GRAPHS (HPCG), which is the problem of deciding whether a given cubic graph admits a Hamiltonian path~\cite{akiyama1980np}.

%We start with some additional definitions and some technical lemmas.
%
%Let $F$ be a footprint such that the number of squares horizontally or vertically aligned is at most three. One such footprint is called \emph{3-aligned footprint}. Let $F^*$ be a graph with a vertex for each square in $F$ and an edge between two squares if and only if the two squares are horizontally or vertically aligned. $F^*$ is called the \emph{supporting graph} of $F$. 
%
%\begin{lemma}
%Let $F$ be a 3-aligned footprint with $n$ squares. A path with $n$ vertices admits a \gbr{} with footprint $F$ if and only if the supporting graph $F^*$ of $F$ admits a Hamiltonian path.
%\end{lemma}
%\begin{proof}
%Since there are at most three squares horizontally and vertically aligned, we can choose a height for each square of $F$ so that the resulting \bgr{} has all edges of $F^*$.  
%\end{proof}

\begin{theorem}\label{th:hardness}
Deciding whether a given graph $G$ admits a \gbr{} with a given footprint is NP-complete, even if $G$ is a path. 
\end{theorem}
\begin{sketch}
We first prove that 2.5GBR-WGF is in NP. A candidate solution consists of a mapping of the vertices of $G$ to the squares of the given footprint and a choice of the heights of the boxes. By Lemma~\ref{le:heights} we can assign to each box an integer height in the set $\{1,2,\dots,n\}$. Thus the size of a candidate solution is polynomial in the size of the input graph. Given a candidate solution, we can test in polynomial time whether all edges of $G$ are realized as visibilities. Thus, the problem is in NP. 

We now describe a reduction from the HPCG problem. Let $G_H$ be an instance of the HPCG problem, i.e. a cubic graph, with $n_H$ vertices and $m_H$ edges. We compute an orthogonal grid drawing $\Gamma_H$ of $G_H$ such that every edge has exactly one bend and no two vertices share the same $x$- or $y$-coordinate. Such a drawing always exists and can be computed in polynomial time with the algorithm by Bruckdorfer et al.~\cite{JGAA-316}. We now use $\Gamma_H$ as a trace to construct an instance $\langle G, F \rangle$ of the 2.5GBR-WGF problem, where $G$ is a path and $F$ a footprint, i.e, a set of squares. $G$ is a path with $4n_H+m_H$ vertices and therefore $F$ will contain $4n_H+m_H$ squares. The footprint $F$ is constructed as follows. $\Gamma_H$ is scaled up by a factor of four. In this way, every two vertices/bends are separated by at least four grid units. Each vertex $v$ of $\Gamma_H$ is replaced by a set $S(v)$ of four unit squares. In particular if vertex $v$ has coordinates $(4x,4y)$ in $\Gamma_H$, then it is replaced by the following four unit squares: $S_1(v)$ whose bottom-right corner has coordinates $(4x,4y)$, $S_2(v)$ whose bottom-right corner has coordinates $(4x+2,4y)$, $S_3(v)$ whose bottom-right corner has coordinates $(4x,4y-2)$, and $S_4(v)$ whose bottom-right corner has coordinates $(4x+2,4y-2)$. We associate with each edge $e$ incident to a vertex $v$, one of the four squares in $S(v)$. If $e$ enters $v$ from West, North, South, or East, the square associated with $e$ is $S_1(v)$, $S_2(v)$, $S_3(v)$, or $S_4(v)$, respectively. Let $(u,v)$ be an edge of $\Gamma_H$ and let $S_i(u)$ and $S_j(v)$ ($1 \leq i,j \leq 4$) be the squares associated with $(u,v)$. The bend of $e=(u,v)$ is replaced by a unit square $S_{e}$ horizontally/vertically aligned with $S_i(u)$ and $S_j(v)$. The set of squares replacing the vertices of $\Gamma_H$, which will be called \emph{vertex squares} in the following, together with the set of squares replacing the bends, which will be called \emph{edge squares} in the following, form the footprint $F$.  Figure~\ref{fi:footprint} shows an orthogonal drawing of a cubic graph and the corresponding footprint $F$. Observe that the footprint $F$ is such that any two squares are separated by at least one unit and in each row/column there are at most three squares. 
\begin{figure}[t]
\centering
\begin{subfigure}[b]{.3\linewidth}
\centering
\includegraphics[width=\columnwidth,page=1]{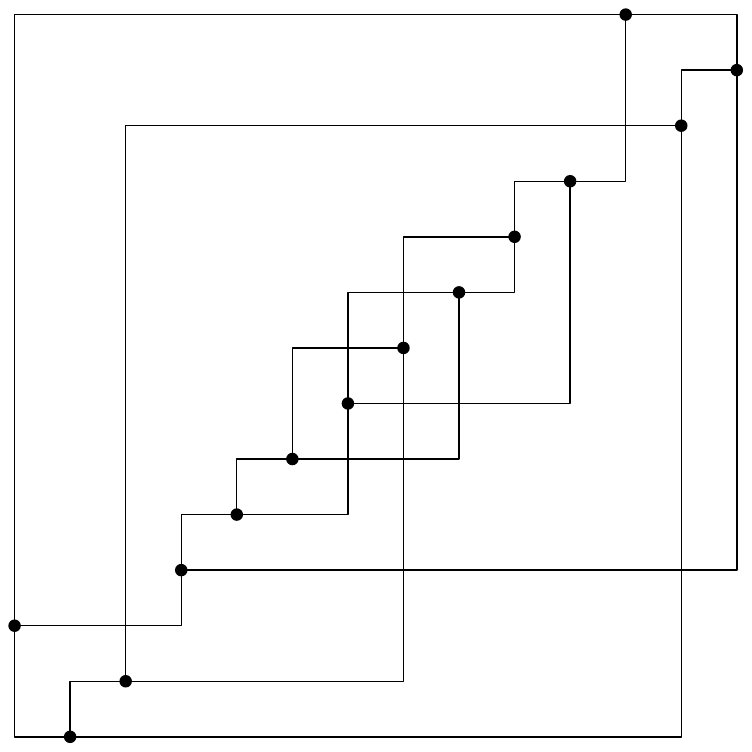}
\caption{}\label{fi:footprint-1}
\end{subfigure}%
\hfil
\begin{subfigure}[b]{.3\linewidth}
\centering
\includegraphics[width=\columnwidth,page=2]{footprint}
\caption{}\label{fi:footprint-2}
\end{subfigure}
\hfil
\begin{subfigure}[b]{.3\linewidth}
\centering
\includegraphics[width=\columnwidth,page=3]{footprint}
\caption{}\label{fi:footprint-3}
\end{subfigure}
\caption{(a) An orthogonal drawing of a cubic graph. (b) Construction of the footprint. Black (gray) squares are vertex (edge) squares. (c) The constructed footprint.}\label{fi:footprint}
\end{figure}
Let $F^*$ be a graph with a vertex for each square in $F$ and an edge between two squares if and only if the two squares are horizontally or vertically aligned. It can be proved that $G_H$ admits a Hamiltonian path if and only if $F^*$ contains a Hamiltonian path, see Appendix~\ref{ap:spiderman}.

Consider the instance $\langle G, F \rangle$ of the 2.5GBR-WGF problem, where $G$ is a path. We prove that $G$ admits a \gbr{} with footprint $F$ if and only if $F^*$ admits a Hamiltonian path. Every graph that can be represented by a \gbr{} with footprint $F$ is a spanning subgraph of $F^*$ (because $F^*$ has all possible edges that can be realized as visibilities in a \gbr{} with footprint $F$). Thus, if $G$ admits a \gbr{} with footprint $F$, then $G$ is a Hamiltonian path of $F^*$ (recall that $G$ is a path). 
Suppose now that $F^*$ has a Hamiltonian path $H^*$. We show that we can choose the heights of the squares in $F$ so that the resulting boxes form a \gbr{} of $G$. Recall that in each row/column of $F$ there are at most three squares. If an edge connects two squares that are consecutive along a row or column, then any choice of the heights is fine. If an edge connects the first and the last square of a row/column, then the heights of these two squares must be larger than the height of the square in the middle. We assign the heights to one square per step, in the order in which they appear along $H^*$. We assign to the first square a height equal to the number of squares (i.e., $4n_H+m_H$). Let $h$ be the height assigned to the current square $S$ and let $S'$ be the next square along $H^*$. If $S$ and $S'$ are consecutive along a row/column then the height assigned to $S'$ is $h$. If $S$ and $S'$ are the first and the last square of a row/column then the height assigned to $S'$ is $h$. If $S$ is the first/last square of a row/column and $S'$ is the middle square of the same row/column, then the height assigned to $S'$ is $h-1$. If $S$ is the middle square of a row/column and $S'$ is the first/last square of the same row/column, then the height assigned to $S'$ is $h+1$. It is easy to see that all heights are positive and that if an edge connects the first and the last square of a row/column, then the heights of these two squares are greater than the height of the square in the middle. This concludes the proof that $G$ admits a \gbr{} with footprint $F$ if and only if $F^*$ admits a Hamiltonian path. Since $F^*$ has a Hamiltonian path if and only if $G_H$ has a Hamiltonian path, $G$ admits a \gbr{} with footprint $F$ if and only if $G_H$ has a Hamiltonian path, which implies that the 2.5GBR-WGF problem is NP-hard.\qed
\end{sketch}

%-------------------------------------------------

\section{Open Problems}
\label{se:open}

There are several possible directions for further study of \br{}s. Among them:
$(i)$ Study the complexity of deciding if a given graph admits a \br{}. We remark that deciding if a graph admits an RVR is NP-hard.
$(ii)$ Investigate other classes of graphs that admit a \br{}. For example, do $1$-planar graphs or partial $5$-trees always admit a \br{}? We remark that there are both $1$-planar graphs and partial $5$-trees not admitting an RVR.
$(iii)$ Study the \br{}s under the strong visibility model. For example, which bipartite graphs admit a strong \br{}?

%\clearpage

%\bibliography{visrepshort,biblio}
% put all the references in 1 file
\bibliography{visrepshort}

\begin{thebibliography}{10}

\bibitem{6777827}
M.~E. Ahmed, A.~B. Yusuf, and M.~Z.~H. Polin.
\newblock Bar 1-visibility representation of optimal 1-planar graph.
\newblock In {\em Electrical Inf. and Comm. Technology (EICT), 2013}, pages
  1--5, 2014.

\bibitem{akiyama1980np}
T.~Akiyama, T.~Nishizeki, and N.~Saito.
\newblock {NP}-completeness of the hamiltonian cycle problem for bipartite
  graphs.
\newblock {\em J. of Inf. Processing}, 3(2):73--76, 1980.

\bibitem{BernhartKainen79}
F.~Bernhart and P.~C. Kainen.
\newblock The book thickness of a graph.
\newblock {\em J. of Combinatorial Theory, Series B}, 27:320--331, 1979.

\bibitem{SoCG}
Therese Biedl, Giuseppe Liotta, and Fabrizio Montecchiani.
\newblock On visibility representations of non-planar graphs.
\newblock In S\'andor Fekete and Anna Lubiw, editors, {\em {SoCG} 2016},
  volume~51 of {\em LIPIcs}, pages 19:1--19:16. Schloss Dagstuhl -
  Leibniz-Zentrum fuer Informatik, 2016.

\bibitem{b-ltaftdst-96}
H.~L. Bodlaender.
\newblock A linear-time algorithm for finding tree-decompositions of small
  treewidth.
\newblock {\em SIAM J. Computing}, 25(6):1305--1317, 1996.

\bibitem{Bodlaender1996358}
H.~L. Bodlaender and T.~Kloks.
\newblock Efficient and constructive algorithms for the pathwidth and treewidth
  of graphs.
\newblock {\em J. of Algorithms}, 21(2):358--402, 1996.

\bibitem{DBLP:conf/gd/BoseDHS96}
P.~Bose, A.~M. Dean, J.~P. Hutchinson, and T.~C. Shermer.
\newblock On rectangle visibility graphs.
\newblock In S.~North, editor, {\em Proc. Symp. on Graph Drawing, {GD} '96},
  volume 1190 of {\em LNCS}, pages 25--44, 1996.

\bibitem{JGAA-6}
P.~Bose, H.~Everett, S.~P. Fekete, M.~E. Houle, A.~Lubiw, H.~Meijer,
  K.~Romanik, G.~Rote, T.~C. Shermer, S.~Whitesides, and C.~Zelle.
\newblock A visibility representation for graphs in three dimensions.
\newblock {\em J. Graph Algorithms and Applications}, 2(3):1--16, 1998.

\bibitem{JGAA-330}
F.~{Brandenburg}.
\newblock 1-visibility representations of 1-planar graphs.
\newblock {\em J. of Graph Algorithms and Applications}, 18(3):421--438, 2014.

\bibitem{JGAA-316}
T.~{Bruckdorfer}, {M.} {Kaufmann}, and {F.} {Montecchiani}.
\newblock 1-bend orthogonal partial edge drawing.
\newblock {\em J. of Graph Algorithms and Applications}, 18(1):111--131, 2014.

\bibitem{Carmi2015251}
Paz Carmi, Eran Friedman, and Matthew~J. Katz.
\newblock Spiderman graph: Visibility in urban regions.
\newblock {\em Computational Geometry}, 48(3):251 -- 259, 2015.

\bibitem{DBLP:journals/combinatorics/DeanEHP08}
A.~M. Dean, J.~A. Ellis{-}Monaghan, S.~Hamilton, and G.~Pangborn.
\newblock Unit rectangle visibility graphs.
\newblock {\em Electr. J. Comb.}, 15(1), 2008.

\bibitem{JGAA-136}
A.~M. {Dean}, {W.} {Evans}, {E.} {Gethner}, {J.}~D. {Laison}, {Md. A.}
  {Safari}, and {W.}~T. {Trotter}.
\newblock Bar $k$-visibility graphs.
\newblock {\em J. of Graph Algorithms and Applications}, 11(1):45--59, 2007.

\bibitem{DBLP:journals/dam/DeanH97}
A.~M. Dean and J.~P. Hutchinson.
\newblock Rectangle-visibility representations of bipartite graphs.
\newblock {\em Disc. App. Math.}, 75(1):9--25, 1997.

\bibitem{JGAA-11}
A.~M. {Dean} and J.~P. {Hutchinson}.
\newblock Rectangle-visibility layouts of unions and products of trees.
\newblock {\em J. of Graph Algorithms and Applications}, 2(8):1--21, 1998.

\bibitem{dv-ubvg-03}
A.~M. Dean and N.~Veytsel.
\newblock Unit bar-visibility graphs.
\newblock {\em Congressus Numerantium}, 160:161--175, 2003.

\bibitem{ddelmmw-opvreg-16}
E.~{Di Giacomo}, W.~Didimo, W.~S. Evans, G.~Liotta, H.~Meijer, F.~Montecchiani,
  and S.~K. Wismath.
\newblock Ortho-polygon visibility representations of embedded graphs.
\newblock {\em CoRR}, abs/1604.08797, 2016.

\bibitem{Duchet1983319}
P.~Duchet, Y.~Hamidoune, M.~Las Vergnas, and H.~Meyniel.
\newblock Representing a planar graph by vertical lines joining different
  levels.
\newblock {\em Discrete Math.}, 46(3):319--321, 1983.

\bibitem{JGAA-343}
{W.} {Evans}, {M.} {Kaufmann}, {W.} {Lenhart}, {T.} {Mchedlidze}, and {S.}
  {Wismath}.
\newblock Bar 1-visibility graphs and their relation to other nearly planar
  graphs.
\newblock {\em J. of Graph Algorithms and Applications}, 18(5):721--739, 2014.

\bibitem{Fekete1995}
S.~P. Fekete, M.~E. Houle, and S.~Whitesides.
\newblock New results on a visibility representation of graphs in {3D}.
\newblock In F.~Brandenburg, editor, {\em Graph Drawing}, volume 1027 of {\em
  LNCS}, pages 234--241. Springer Berlin Heidelberg, 1995.

\bibitem{DBLP:journals/ijcga/FeketeM99}
S.~P. Fekete and H.~Meijer.
\newblock Rectangle and box visibility graphs in {3D}.
\newblock {\em Int. J. Comput. Geometry Appl.}, 9(1):1--28, 1999.

\bibitem{JGAA-157}
{S.} {Felsner} and {M.} {Massow}.
\newblock Parameters of bar $k$-visibility graphs.
\newblock {\em J. of Graph Algorithms and Applications}, 12(1):5--27, 2008.

\bibitem{GareyJS76}
M.~R. Garey, D.~S. Johnson, and H.~C. So.
\newblock An application of graph coloring to printed circuit testing.
\newblock {\em IEEE Trans. on Circuits and Systems}, CAS-23(10):591--599, 1976.

\bibitem{Gupta2005242}
A.~Gupta, N.~Nishimura, A.~Proskurowski, and P.~Ragde.
\newblock Embeddings of $k$-connected graphs of pathwidth $k$.
\newblock {\em Disc. Applied Mathematics}, 145(2):242--265, 2005.

\bibitem{Hutchinson1999161}
J.~P. Hutchinson, T.~Shermer, and A.~Vince.
\newblock On representations of some thickness-two graphs.
\newblock {\em Comp. Geometry}, 13(3):161--171, 1999.

\bibitem{DBLP:journals/ipl/KantLTT97}
G.~Kant, G.~Liotta, R.~Tamassia, and I.~G. Tollis.
\newblock Area requirement of visibility representations of trees.
\newblock {\em Inf. Process. Lett.}, 62(2):81--88, 1997.

\bibitem{Kleitman2001}
J.~D. Kleitman, A.~Gy{\'a}rf{\'a}s, and G.~T{\'o}th.
\newblock Convex sets in the plane with three of every four meeting.
\newblock {\em Combinatorica}, 21(2):221--232, 2001.

\bibitem{DBLP:journals/ipl/LiottaM16}
G.~Liotta and F.~Montecchiani.
\newblock L-visibility drawings of {IC}-planar graphs.
\newblock {\em Inf. Process. Lett.}, 116(3):217--222, 2016.

\bibitem{ov-grild-78}
R.~H. J.~M. Otten and J.~G.~Van Wijk.
\newblock Graph representations in interactive layout design.
\newblock In {\em {IEEE ISCSS}}, pages 91--918. IEEE, 1978.

\bibitem{DBLP:journals/dcg/RosenstiehlT86}
P.~Rosenstiehl and R.~E. Tarjan.
\newblock Rectilinear planar layouts and bipolar orientations of planar graphs.
\newblock {\em Discr. {\&} Comput. Geom.}, 1:343--353, 1986.

\bibitem{Shermer96}
T.~C. Shermer.
\newblock On rectangle visibility graphs {III}. external visibility and
  complexity.
\newblock In {\em Cdn. Conf. on Comp. Geometry}, pages 234--239, 1996.

\bibitem{Stola2009}
J.~{\v{S}}tola.
\newblock Unimaximal sequences of pairs in rectangle visibility drawing.
\newblock In I.~G. Tollis and M.~Patrignani, editors, {\em Graph Drawing 2008},
  volume 5417 of {\em LNCS}, pages 61--66. Springer Berlin Heidelberg, 2009.

\bibitem{DBLP:conf/stacs/StreinuW03}
I.~Streinu and S.~Whitesides.
\newblock Rectangle visibility graphs: Characterization, construction, and
  compaction.
\newblock In H.~Alt and M.~Habib, editors, {\em {STACS} 2003}, volume 2607 of
  {\em LNCS}, pages 26--37. Springer, 2003.

\bibitem{Sultana2014}
S.~Sultana, Md.~S. Rahman, A.~Roy, and S.~Tairin.
\newblock Bar 1-visibility drawings of 1-planar graphs.
\newblock In P.~Gupta and C.~Zaroliagis, editors, {\em Proc. Applied
  Algorithms: First Int. Conf. ICAA 2014}, pages 62--76. Springer Int.
  Publishing, 2014.

\bibitem{TamassiaTollis86}
R.~Tamassia and I.~G. Tollis.
\newblock A unified approach to visibility representations of planar graphs.
\newblock {\em Discr. \& Comput. Geom.}, 1(1):321--341, 1986.

\bibitem{tt-rgc-91}
R.~Tamassia and {I. G.} Tollis.
\newblock Representations of graphs on a cylinder.
\newblock {\em SIAM J. Disc. Mathematics}, 4(1):139--149, 1991.

\bibitem{t-prg-84}
C.~Thomassen.
\newblock Plane representations of graphs.
\newblock In {\em Progress in Graph Theory}, pages 43--69. AP, 1984.

\bibitem{DBLP:conf/compgeom/Wismath85}
S.~K. Wismath.
\newblock Characterizing bar line-of-sight graphs.
\newblock In {\em Proc. 1st Symp. on Comp. Geometry}, pages 147--152, 1985.

\end{thebibliography}
\bibliographystyle{plain}

\newpage

\appendix
\section*{Appendix}\label{se:appendix}

\section{Additional Material for Section~\ref{se:complete}}\label{ap:complete}

%\paragraph{Proof of Theorem~\ref{th:bipartite}.} 
%Let $K_{m,n}$ be a complete bipartite graph. We represent the $m$ vertices in the first partite set with $m$ boxes $a_0, a_1,$ $\ldots, a_{m-1}$ such that box $a_i$ has a footprint with corners at ($2i,0,0$),  ($2i+1,0,0$),  ($2i,2n-1,0$) and  ($2i+1,2n-1,0$) and height $m-i$.
%Then we represent the $n$ vertices in the second partite set with $n$ boxes $b_0,b_1,\ldots,b_{n-1}$ such that box $b_j$ has a footprint with corners at ($2m,2j,0$),  ($2m+1,2j,0$),  ($2m,2j+1,0$) and ($2m+1,2j+1,0$) and height $m$. Consider now a box $a_i$ and a box $b_j$. By construction $a_i$ and $b_j$ see each other above all boxes $a_l$ with $l>i$.\qed

\paragraph{Proof of Lemma~\ref{le:heights}.} 
By hypothesis $G$ admits a \br{} $\Gamma'$. If every box of $\Gamma'$ has a distinct integer height in the range $[1,n]$, the statement is true. If not, we can change the heights so to achieve this condition. Namely, denote by $b_1,b_2,\dots,b_n$ the boxes of $\Gamma'$ in non-decreasing order of height; we change the height of $b_i$ to be $i$ (for $i=1,2,\dots,n$). Let $\Gamma$ be the resulting representation and denote by $h'(b_i)$ the height of $b_i$ in $\Gamma'$ and by $h(b_i)$ the height of $b_i$ in $\Gamma$. For any two boxes $b_i$ and $b_j$, $h(b_i) < h(b_j)$ if and only if $h'(b_i) \leq h'(b_j)$, which means that no visibility has been destroyed by our change of the heights (while some new visibility may have been created).\qed

\paragraph{Proof of Lemma~\ref{le:cross}.} 
For a given arrangement $\R$, choose $\ell_v$ and $\ell_h$ to be a
vertical and horizontal line whose union intersects the maximum number
of rectangles in $\R$.  Suppose, for the sake of contradiction, that
some rectangle $a \in \R$ is not intersected by $\ell_v \cup \ell_h$.
Choose $\ell_v$ and $\ell_h$ so that they are closest to $a$ without
changing the set of rectangles intersected by their union.
Assume w.l.o.g. that $a$ lies in the positive quadrant of $\ell_v \cup
\ell_h$.
Let $b$ be a rectangle that prevents $\ell_h$ from moving closer to
$a$, that is, $b \cap \ell_h \neq \emptyset$ but $b \cap (\ell_v \cup
\ell'_h) = \emptyset$, where $\ell'_h$ is $\ell_h$ translated in the $+y$ direction by any arbitrarily
small positive amount. Let $c$ be a rectangle that prevents $\ell_v$ from moving closer to
$A$, that is, $c \cap \ell_v \neq \emptyset$ but $c \cap (\ell'_v \cup
\ell_h) = \emptyset$, where $\ell'_v$ is $\ell_v$ translated in the $+x$ direction by any arbitrarily
small positive amount.
The line $\ell_h$ separates $a$ and $b$, so $y(a) \cap y(b) =
\emptyset$, which implies $x(a) \cap x(b) \neq \emptyset$.
Similarly, using line $\ell_v$, $y(a) \cap y(c) \neq \emptyset$.

By the conditions of the lemma, either $y(b) \cap y(c) \neq \emptyset$
or $x(b) \cap x(c) \neq \emptyset$.
Suppose that $y(b) \cap y(c) \neq \emptyset$.
Since $y(c)$ has non-empty intersection with both $y(a)$ and $y(b)$,
any horizontal line that separates $a$ and $b$ must intersect $c$.
Thus $\ell_h$ intersects $c$ and $c \cap (\ell'_v \cup \ell_h) \neq
\emptyset$ for all vertical lines $\ell'_v$; a contradiction with the fact that $c$ prevents $\ell_v$ from moving closer to $a$.
We obtain a similar contradiction if $x(b) \cap x(c) \neq \emptyset$.\qed

\begin{lemma}\label{le:k19}
$K_{19}$ admits a \br{}.
\end{lemma}
\begin{proof}
Let $\Lambda$ be the box collection shown in Figure~\ref{fig:K19}, where the footprint $\Lambda_0$ is depicted by a 2D drawing, while the heights of boxes are indicated by integer labels. We prove the statement by showing that $\Lambda$ is a \br{} of $K_{19}$. We preliminarily observe that $\Lambda_0$ satisfies the necessary condition \nc{}. The boxes of $\Lambda$ are partitioned into five subsets $T, R, L, B, C$ such that $|T| = |R| = 4$, $|L| = |B| = 3$ and $|C| = 5$; this partitioning is shown by the five rectangles with thick sides. For $S \in \{ T,R,L,B \}$, it is easy to see that (for any assignment of heights to the boxes in $\Lambda$) $\Lambda(S)$ is a \br{} of $K_{|S|}$ in $\Lambda$, since any two boxes of $\Lambda(S)$ are ground visible in $\Lambda_0$.  Hence, except for $C$, the intra-partition visibilities are ensured regardless of the heights of the boxes. We now show that even the inter-partition visibilities and the intra-partition visibilities of $\Lambda(C)$ exist if the heights of boxes are chosen as shown in Figure~\ref{fig:K19}. We follow an incremental strategy in which, at each step, we add one of $R, L ,B, C$, in order, to the current set of vertices, which is initialized to $T$.

%Step~1
\emph{Step~1: addition of $R$}. Box $\Lambda(t_i)$ ($1 \leq i < |T|$) may obstruct the visibility between a box in $\Lambda(R)$ and a box $\Lambda(t_{i'})$, with $i' > i$, in $\Lambda(T)$.
This obstruction can be avoided, however, if (\emph{i}) $\Lambda(T)$ has a staircase layout, i.e. the height of its boxes increases as $\Lambda(R)$ gets farther, and (\emph{ii}) every box in $\Lambda(R)$ is not lower than any box in $\Lambda(T)$. Therefore, $\Lambda(T \cup R)$ is a \br{} of $K_8$ if $h(t_i) = i$ ($1 \leq i \leq |T|$) and, for each $1 \leq j \leq |R|$, $h(r_j) \geq h_{max}(T)$, where $h_{max}(T)$ is the maximum box height in $\Lambda(T)$.

%Step~2
\emph{Step~2: addition of $L$}. Consider the subset $S' = L \cup \{r_1\}$. As for the previous partite sets, $\Lambda(S')$ is a \br{} of $K_{|S'|}$ in $\Lambda$ for any assignment of heights.
However, box $\Lambda(r_i)$ ($1 \leq i < |R|$) may prevent the (inter-partition) visibility between a box in $\Lambda(L)$ and a box $\Lambda(r_{i'})$, with $i' > i$, in $\Lambda(R)$.
As before, these visibilities can be ensured if $\Lambda(R)$ has a staircase layout and every box in $\Lambda(L)$ is not lower than any box in $\Lambda(R)$.
In particular, a possible assignment of heights is the following: $h(r_i) = h_{max}(T) + i - 1$ ($1 \leq i \leq |R|$) and, for each $1 \leq j \leq |L|$, $h(l_j) \geq h_{max}(R)$.
This assignment also implies the visibility between every box in $\Lambda(L)$ and every box in $\Lambda(T)$.

%Step~3
\emph{Step~3: addition of $B$}. Consider now the subset $S'' = B \cup \{l_1\}$.
As before, $\Lambda(S'')$ is a \br{} of $K_{|S''|}$ in $\Lambda$, independently from the choice of the heights of boxes.
Furthermore, the inter-partition visibilities between $\Lambda(B)$ and $\Lambda(L)$ can be satisfied if $h(l_i) = h_{max}(R) + i - 1$ ($1 \leq i \leq |L|$) and, for each $1 \leq j \leq |B|$, $h(b_j) \geq h_{max}(L)$.
With this assignment of heights, the inter-partition visibilities between $\Lambda(B)$ and $\Lambda(T)$ and between $\Lambda(B)$ and $\Lambda(R)$ are ensured. 

%Step~4
\emph{Step~4: addition of $C$}. According to Fig.~\ref{fig:K19}, every box $\Lambda(c_i)$ ($1 \leq i \leq 5$) can see every other box $\Lambda(v)$ with $v \in T \cup R \cup L \cup B$, provided that $h(c_i) \geq h(v)$.
Indeed, $\Lambda(c_i)$ is ground visible to any box in $\{\Lambda(t_1), \Lambda(r_1), \Lambda(b_1), \Lambda(l_1)\}$ and $\Lambda(T)$, $\Lambda(R)$, $\Lambda(L)$, $\Lambda(B)$ have a staircase layout with increasing box height as $\Lambda(C)$ gets farther.
Therefore, if $h(c_i) \geq h_{max}(B)$ ($1 \leq i \leq 5$), then all the inter-partition visibilities in $\Lambda$ are satisfied.
It remains to consider the intra-partition visibilities in $\Lambda(C)$.
In this regard, the only visibility obstructions can be caused by $\Lambda(c_5)$, which may prevent the visibility between $\Lambda(c_1)$ and $\Lambda(c_3)$ ($\Lambda(c_2)$ and $\Lambda(c_3)$, respectively) if $h(c_1)$ and $h(c_3)$ ($h(c_2)$ and $h(c_4)$, respectively) are not both strictly greater than $h(c_5)$.
Therefore, by choosing $h(c_5) = h_{max}(B)$ and $h(c_i) = h(c_5) + 1$ ($1 \leq i \leq 4$), the intra-partition visibilities in $\Lambda(C)$ are satisfied, from which it follows that $\Lambda$ is a \br{} of $K_{19}$.\qed
\end{proof}

\section{Additional Material for Section~\ref{se:pathwidth}}\label{ap:pathwidth}

A \emph{path decomposition} $P$ of a graph $G=(V,E)$ is a sequence $P_1, \dots, P_k$ of subsets of $V$, called \emph{bags}, such that the following three properties hold:
\begin{itemize}
\item For every vertex $u$ of $G$, there is a bag $P_i$ (with $1 \leq i \leq k$) such that $u \in P_i$;
\item For every edge $(u,v)$ of $G$, there is a bag $P_i$ (with $1 \leq i \leq k$) such that $u,v \in P_i$;
\item For every vertex $u$, there exists two indices $1 \leq j \leq h \leq k$, such that $u$ is contained in all bags $P_i$ such that $j \leq i \leq h$ and in no other bag.
\end{itemize}

Let $P_i$ be the bag of $P$ with maximum size. The \emph{width} of the path decomposition $P$ is $|P_i|-1$. The \emph{pathwidth} of a graph $G$ is the minimum width of any path decomposition of $G$. 

A path decomposition $P=P_1,\dots,P_k$ of a graph $G$ of pathwidth $p$ is normalized if $|P_i|=p+1$ for $i$ odd, $|P_i|=p$ for $i$ even, and $P_{i-1} \cap P_{i+1}=P_i$ for $i$ even.

For a fixed $p$, path decomposition of graphs with pathwidth $p$ can be found in linear time~\cite{Bodlaender1996358,b-ltaftdst-96}. Given a path decomposition, a normalized path decomposition of the same width can be found in linear time~\cite{Gupta2005242}.

\paragraph{Missing cases of the proof of Theorem~\ref{th:pathwidth}.} 
To complete the proof of Theorem~\ref{th:pathwidth}, we need to prove the following cases.

\begin{description}

\item[$h=3$ or $h=4$.]  The box $\Lambda_i(v_i)$ only intersects the regions $\alpha_{4,k'}$, with $k' <4$, and $\alpha_{k'',4}$, with $k'' > 4$. Thus, the only visibilities that could be destroyed are those inside these regions. The visibilities inside $\alpha_{4,1}$ and $\alpha_{4,2}$ are not destroyed because $r_{h,i}$ is placed so that the existing boxes of groups $4$ can still see (to the left of $r_{h,i}$) the boxes of group $1$ and $2$ inside $\alpha_{4,1}$ and $\alpha_{4,2}$. The existing visibilities between boxes of group $4$ and the boxes of group $3$ are not destroyed because $r_{h,i}$ is short enough (in the $x$-direction) so that the existing boxes of groups $3$ and $4$ can still see each other (to the right of $r_{h,i}$) inside $\alpha_{4,3}$. The visibilities between the vertices of group $4$ and vertices of group $h'$, with $h'>4$, are not destroyed because the boxes representing the vertices of group $4$ are taller than $\Lambda_i(v_i)$ and so are the boxes of any group $h'$ with $h'>4$. So, no visibility is destroyed for the vertices of group $4$. The box $\Lambda_i(v_i)$ sees the active vertex of group $1$ with a visibility that is inside $\alpha_{3,1}$ or $\alpha_{4,1}$ and above the boxes of group $1$ corresponding to non-active vertices (these boxes are shorter than the box of the active vertex of group one).    
Similarly, $\Lambda_i(v_i)$ sees the active vertex of group $2$ with a visibility that is inside $\alpha_{3,2}$ or $\alpha_{4,2}$ and above the boxes of group $1$ (including the active one). The box $\Lambda_i(v_i)$ sees the active vertex of group $3$ or $4$ via a ground visibility inside $\alpha_{4,3}$ and it sees the boxes of all the other active vertices inside $\alpha_{h,k}$ (with $h=3$ or $4$, and $k > 4$) above the boxes of group $3$ (which are all shorter than it).

\item[$h=7$ or $h=8$.]  The box $\Lambda_i(v_i)$ only intersects the regions $\alpha_{7,5}$, $\alpha_{8,5}$, and $\alpha_{8,7}$. Thus, the only visibilities that could be destroyed are those inside these regions. The visibilities between the vertices of groups $5$ and $7$ are not destroyed because $r_{h,i}$ is placed so that the existing boxes of groups $7$ can still see the boxes of group $5$ inside $\alpha_{7,5}$ above $r_{h,i}$ (in the $y$-direction). Similarly, the visibilities between the vertices of group $5$ and the vertices of group $8$ are not destroyed because the existing boxes of groups $8$ can still see the boxes of group $5$ inside $\alpha_{8,5}$ and below $r_{h,i}$ (in the $y$-direction). The visibilities between the vertices of group $7$ and the vertices of group $8$ are not destroyed because the existing boxes of groups $8$ can still see the boxes of group $7$ inside $\alpha_{8,7}$ and above $r_{h,i}$ (in the $y$-direction), if $h=8$, or below $r_{h,i}$ (in the $y$-direction), if $h=7$. The proof that $v_i$ sees all the other active vertices is equal to the one in the case $h=5$ or $h=6$.\qed
\end{description}

\section{Additional Material for Section~\ref{se:spiderman}}\label{ap:spiderman}

\paragraph{Proof of Theorem~\ref{th:density}.} 
Consider a \gbr{} $\Gamma$ of an $n$-vertex graph $G$ and let $\Gamma_0$ be the footprint of $\Gamma$. $\Gamma_0$ consists of a set of unit squares that are aligned along the columns and the rows of an integer grid. Let $B_1, B_2, \dots, B_k$ be a set of boxes whose footprints share the same $x$-extent (i.e., are in the same column) or $y$-extent (i.e., are in the same row). Two boxes $B_i$ and $B_j$, with $i+1<j$  can have a visibility only if they are both higher than any other box $B_k$ with $ i < k <j$. Thus, there cannot be four boxes $B_i, B_j, B_k, B_h$, with $i < j < k < h$ and such that there is a visibility between $B_i$ and $B_k$ and a visibility between $B_j$ and $B_h$. If these two visibilities existed then $B_j$ should be taller than $B_k$ (in order to see $B_h$) and $B_k$ should be taller than $B_j$ (in order to see $B_i$). It follows that the subgraph $G'$ of $G$ that is represented by the boxes $B_1, B_2, \dots, B_k$ has page number one and therefore is outerplanar\footnote{A graph $G=(V,E)$ has page number one if there exists a total order $\leq$ of $V$ such that there are no two edges $(u,v)$ and $(w,z)$ with $u \leq w \leq v \leq z$. It is known that a graph has page number one if and only if it is outerplanar~\cite{BernhartKainen79}.}. This implies that the maximum number of edges of $G'$ is $2k-3$. Suppose that the unit squares in $\Gamma_0$ occupy $R$ rows and $C$ columns, and that row $r_i$ (with $1 \leq i \leq R$) has $n_{r_i}$ vertices while column $c_i$ (with $1 \leq i \leq C$) has $n_{c_i}$ vertices. The maximum number of edges in $G$ is $\sum_{i=1}^R(2n_{r_i}-3)+\sum_{i=1}^C(2n_{c_i}-3)=2n-3R+2n-3C=4n-3(R+C)$. It is easy to see that this number is maximized when $R=C=\sqrt{n}$, i.e., when the squares of $\Gamma_0$ form a $\sqrt{n} \times \sqrt{n}$ grid. In this case the maximum number of edges is $4n-6\sqrt{n}$.

%-----------------  Add this figure?

\begin{figure}[t]
\centering
\includegraphics[scale=0.6]{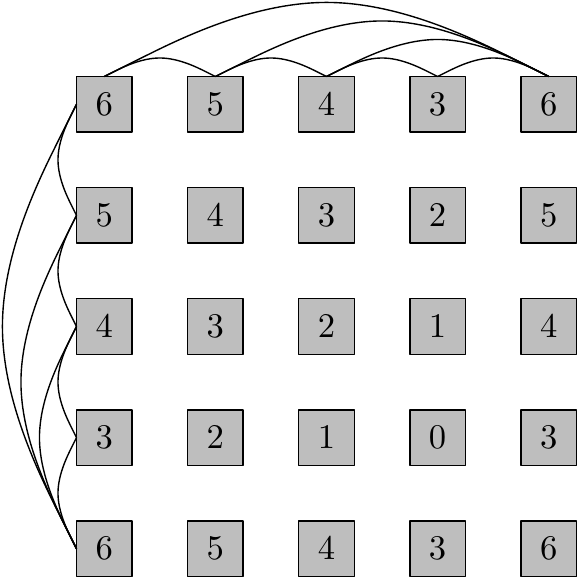}
\caption{A \gbr{} with maximum edge density. The numbers in the square indicate the height of each box. The edges of the subgraph induced by the first row and the first column are shown.}
\label{fig:maxedges}
\end{figure}

For each $n=k^2$ a \gbr{} that achieves the maximum edge density can be created by placing $n$ boxes so that the footprint is a $k \times k$ grid and then choosing the heights of the boxes along each row and column to form a decending sequence between two maxima. Figure~\ref{fig:maxedges} shows a $5\times 5$ grid with appropriate heights indicated. The pattern is easily extended.\qed

\paragraph{Missing part of the proof of Theorem~\ref{th:hardness}.} 
To complete the proof of Theorem~\ref{th:hardness}, we must prove that $G_H$ admits a Hamiltonian path if and only if $F^*$ contains a Hamiltonian path. Recall that $F^*$ denotes a graph with a vertex for each square in $F$ and an edge between two squares if and only if the two squares are horizontally or vertically aligned. Suppose first that $G_H$ has a Hamiltonian path $H$. Each edge $e=(u,v)$ of $H$ corresponds to two edges in $F^*$: one connecting a vertex square $S_i(u)$ ($1 \leq i \leq 4$) to the edge square $S_e$, and one connecting $S_e$ to a vertex square $S_j(v)$ ($1 \leq j \leq 4$). Notice however that the set $E_H$ of such edges does not form a Hamiltonian path of $F^*$ because it does not contain all vertices of $F^*$. In particular, for each vertex $v$ of $G_H$ there are two vertex squares $S_i(v)$ and $S_j(v)$ ($1 \leq i,j \leq 4$) that have no incident edge in $E_H$ (for the end-vertices of $H$, the vertex squares without incident edges in $E_H$ are three); also, the edge squares of all the edges that are not in $H$ have no incident edges  in $E_H$. We say that these edge squares are \emph{orphans}. We now show that it is possible to select additional edges of $F^*$ to create a Hamiltonian path. Each orphan edge square is assigned to one of the end-vertices of the edge corresponding to the square. The assignment is arbitrary, we only take care that at most one orphan edge square is assigned to each end-vertex of $H$. Let $v$ be a vertex of $G_H$ and suppose first that $v$ is an internal vertex of $H$. We have two cases: an orphan edge square is associated with $v$ or not. In both cases then we have two sub-cases: the vertex square of $v$ with an incident edge of $E_H$ are horizontally/vertically aligned or not. Figures~\ref{fi:transformations-a-1}-\ref{fi:transformations-a-4} shows for each sub-case how to select additional edges of $F^*$ so that the four vertex squares of $v$ and, possibly, the orphan edge square assigned to $v$ are traversed by a simple path. 
The case when $v$ is an end-vertex of $H$ can be treated similarly, Figures~\ref{fi:transformations-b-1}-\ref{fi:transformations-b-3} shows the possible cases. It is easy to see that applying the transformations illustrated in Figure~\ref{fi:transformations-ab} to all the vertices $v$ of $G_H$, we obtain a Hamiltonian path of $F^*$. Figure~\ref{fi:transformations-complete} shows a complete example.

\begin{figure}[tbp]
\centering
\begin{subfigure}[b]{.22\linewidth}
\centering
\includegraphics[width=\columnwidth,page=1]{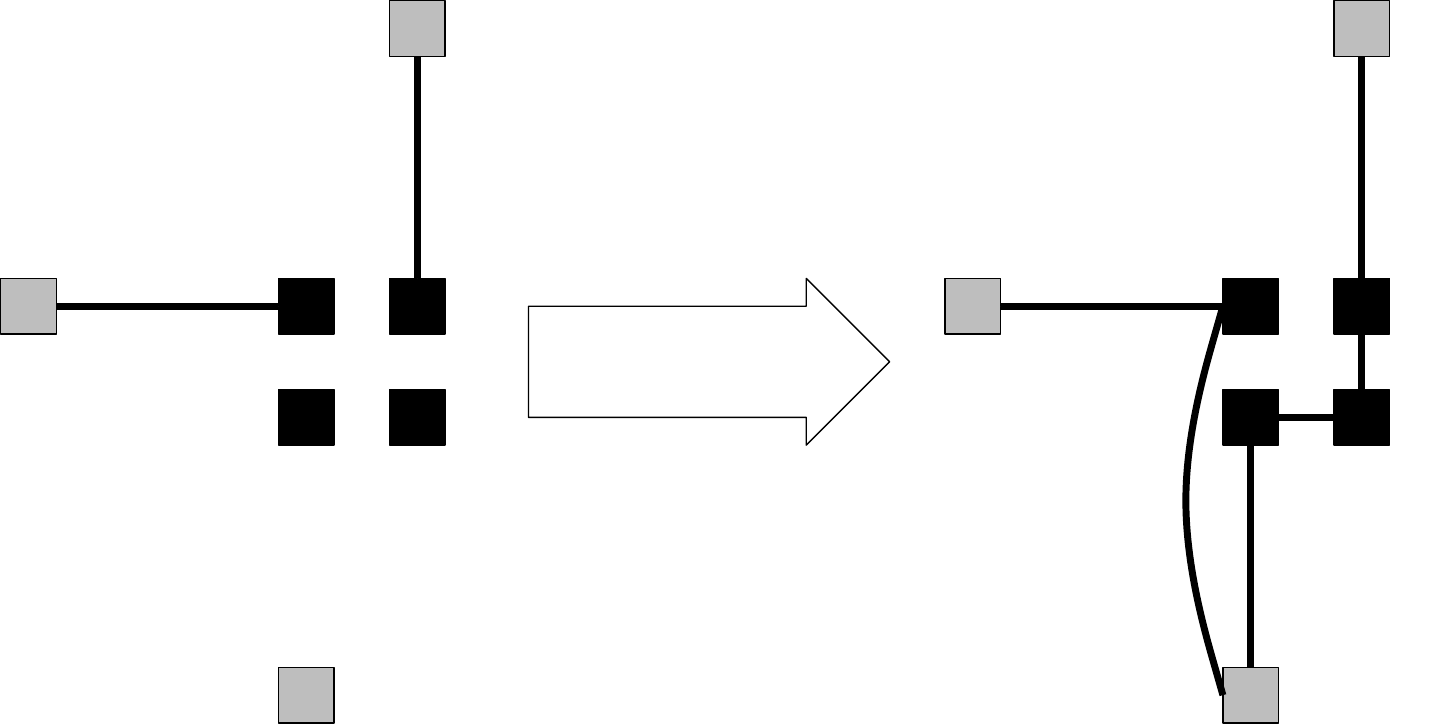}
\caption{}\label{fi:transformations-a-1}
\end{subfigure}%
\hfil
\begin{subfigure}[b]{.22\linewidth}
\centering
\includegraphics[width=\columnwidth,page=2]{transformations-a}
\caption{}\label{fi:transformations-a-2}
\end{subfigure}%
\hfil
\begin{subfigure}[b]{.22\linewidth}
\centering
\includegraphics[width=\columnwidth,page=3]{transformations-a}
\caption{}\label{fi:transformations-a-3}
\end{subfigure}%
\hfil
\begin{subfigure}[b]{.22\linewidth}
\centering
\includegraphics[width=\columnwidth,page=4]{transformations-a}
\caption{}\label{fi:transformations-a-4}
\end{subfigure}%
\\
\begin{subfigure}[b]{.22\linewidth}
\centering
\includegraphics[width=\columnwidth,page=1]{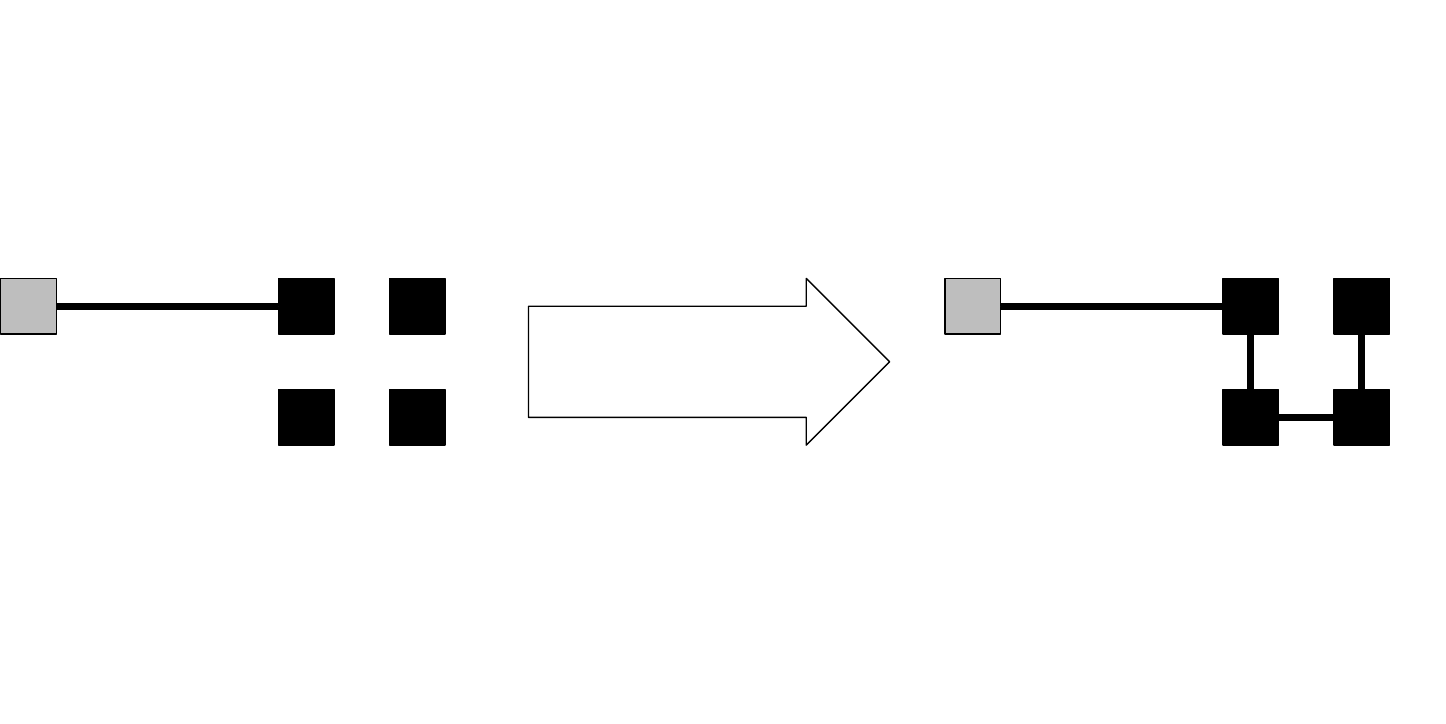}
\caption{}\label{fi:transformations-b-1}
\end{subfigure}%
\hfil
\begin{subfigure}[b]{.22\linewidth}
\centering
\includegraphics[width=\columnwidth,page=2]{transformations-b}
\caption{}\label{fi:transformations-b-2}
\end{subfigure}%
\hfil
\begin{subfigure}[b]{.22\linewidth}
\centering
\includegraphics[width=\columnwidth,page=3]{transformations-b}
\caption{}\label{fi:transformations-b-3}
\end{subfigure}%
\caption{Illustration for the proof of Theorem~\ref{th:hardness}}.\label{fi:transformations-ab}
\end{figure}

\begin{figure}[t]
\centering
\begin{subfigure}[b]{.3\linewidth}
\centering
\includegraphics[width=\columnwidth,page=1]{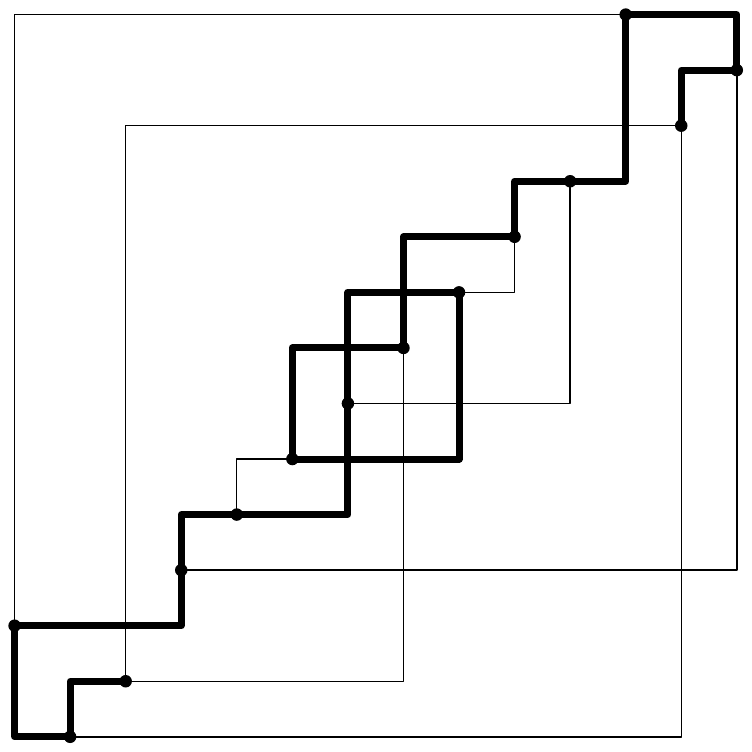}
\caption{}\label{fi:transformations-complete-1}
\end{subfigure}%
\hfil
\begin{subfigure}[b]{.3\linewidth}
\centering
\includegraphics[width=\columnwidth,page=2]{transformations-complete}
\caption{}\label{fi:transformations-complete-2}
\end{subfigure}
\hfil
\begin{subfigure}[b]{.3\linewidth}
\centering
\includegraphics[width=\columnwidth,page=3]{transformations-complete}
\caption{}\label{fi:transformations-complete-3}
\end{subfigure}
\caption{(a) A Hamiltonian path $H$ of the cubic graph $G_H$. (b) Constructing a Hamiltonian path of $F^*$. The bold edges are the edges of $F^*$ that correspond to the edges of $H$. (c) A Hamiltonian path of $F^*$.}\label{fi:transformations-complete}
\end{figure}

\begin{figure}[t]
\centering
\begin{subfigure}[b]{.3\linewidth}
\centering
\includegraphics[width=\columnwidth,page=1]{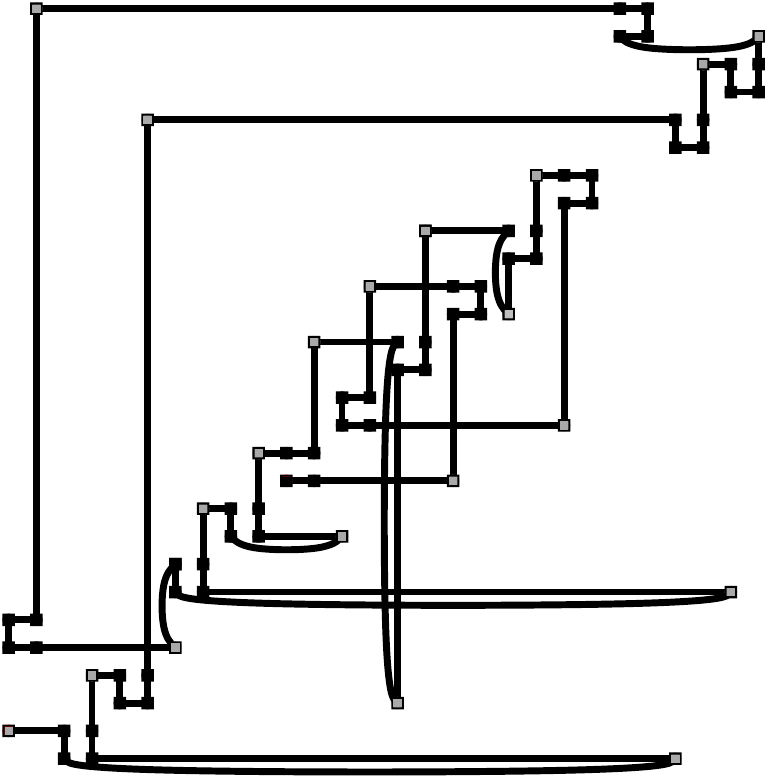}
\caption{}\label{fi:hamiltonian-1}
\end{subfigure}%
\hfil
\begin{subfigure}[b]{.3\linewidth}
\centering
\includegraphics[width=\columnwidth,page=2]{hamiltonian}
\caption{}\label{fi:hamiltonian-2}
\end{subfigure}
\hfil
\begin{subfigure}[b]{.3\linewidth}
\centering
\includegraphics[width=\columnwidth,page=3]{hamiltonian}
\caption{}\label{fi:hamiltonian-3}
\end{subfigure}
\caption{(a) A Hamiltonian path $H^*$ of $F^*$. (b) $H^*$ simplified. (c) The candidate edges $E_H$ of $G_H$. The dashed edge has to be removed to create a Hamiltonian path $H$ of $G_H$.}\label{fi:hamiltonian}
\end{figure}

Suppose now that $F^*$ has a Hamiltonian path $H^*$. We show how to construct a Hamiltonian path $H$ of $G_H$. We first make a simplification of $H^*$. Let $S_e$ be an edge square of $F^*$ and suppose that it is adjacent in $H^*$ to two vertex squares $S_i(v)$ and $S_j(v)$ of the same vertex $v$. In this case $S_e$ and its adjacent edges are removed from $H^*$ and the edge connecting $S_i(v)$ and $S_j(v)$ is added to $H^*$. Analogously, if $S_e$ is an end-vertex of $H^*$, we remove it (and its adjacent edge) from $H^*$. After this simplification, $H^*$ is a simple path that traverses all vertex squares and a subset of the edge squares. Each edge square $S_e$ that is still in $H^*$ is adjacent to two vertex squares $S_i(u)$ and $S_j(v)$ of two different vertices $u$ and $v$. Edge $(u,v)$ of $G_H$ is selected as a candidate edge for $H$. We now show that the set $E_H$ of candidate edges forms a Hamiltonian path $H$ of $G_H$, possibly after the removal of one or two edges. Let $v$ be a vertex of $G_H$ such that no end-vertex of $H^*$ is a vertex square of $v$.  We claim that the four vertex squares of $v$ appear consecutively in $H^*$. Since $G_H$ is a cubic graph, there can be at most three edges of $H^*$ with an end-vertex in $S(v)$ and the other end-vertex outside $S(v)$. However, since any other edge of $H^*$ incident to a  square of $S(v)$ must have the other end-vertex also in $S(v)$, if one or three such edges existed, then a square of $S(v)$ would be an end-vertex of $H^*$, but we are assuming that this is not the case. Thus, the four squares of $S(v)$ are connected to two squares not in $S(v)$ and therefore they must be consecutive in $H^*$. This implies that each vertex $v$ of $G_H$, except at most two, has at most two incident candidate edges. The exceptions are the (at most) two vertices whose vertex squares include the end-vertices of $H^*$. This gives rise to two edges of $E_H$ that can be removed to obtain a Hamiltonian path $H$. Figure~\ref{fi:hamiltonian} shows an example.\qed

\end{document}